\documentclass[ecta]{econsocart}
\usepackage{comment,color}
\usepackage{array,tabularx}
\usepackage{arydshln} % Dashed array and tabular lines
\usepackage{booktabs,rotating,multirow,framed}
\newcolumntype{R}{>{\raggedleft\arraybackslash}X}

\numberwithin{equation}{section}
\numberwithin{figure}{section}
\theoremstyle{plain}
\newtheorem{theorem}{Theorem}[section]
\newtheorem{lemma}{Lemma}[section]
\newtheorem{corollary}{Corollary}[section]
\newtheorem{assumption}{Assumption}[section]
\newtheorem*{lemmaA1}{Lemma A.1}
\newtheorem*{lemmaA2}{Lemma A.2}

\theoremstyle{remark}

\startlocaldefs

\providecommand{\e}[1]{\ensuremath{\times 10^{#1}}}
\DeclareMathOperator{\tr}{tr}
\newcommand{\D}{\mathrm{d}}
\newcommand{\E}{\mathrm{e}}
\newcommand{\mye}{\ensuremath{\mathsf{E}}}

\newcommand{\pa}{{\partial}}

\newcommand{\ignore}[1]{}

\newcommand{\cF}{\mathcal{F}}

\newcommand{\tv}{\tilde{v}}
\newcommand{\tu}{\tilde{u}}
\newcommand{\tw}{\tilde{w}}
\newcommand{\tzeta}{\tilde{\zeta}}
\def \a{\alpha}
\def \b{\beta}
\def \e{\varepsilon}

\def \1{\mathbf 1}

\def\P{\mathbb{P}}
\def\Q{\mathbb{Q}}
\def\R{\mathbb{R}}

\DeclareMathOperator*{\argmax}{arg\,max}

\endlocaldefs

\begin{document}

\begin{frontmatter}

\title{Optimal Transport and Risk Aversion \\in Kyle's Model of Informed Trading}
\runtitle{Optimal Transport and Risk Aversion in the Kyle Model}

\def\date{August 11, 2021}

\begin{aug}
% use \particle for den|der|de|van|von (only lc!)
% [id=?,addressref=?,corref]{\fnms{}~\snm{}\ead[label=e?]{}\thanksref{}}
%
%% e-mail is mandatory for each author
%
%%% initials in fnms (if any) with spaces
%
\author[id=au1,addressref={add1}]{\fnms{Kerry}~\snm{Back}\ead[label=e1]{kerry.e.back@rice.edu}}
\author[id=au2,addressref={add2}]{\fnms{Francois}~\snm{Cocquemas}\ead[label=e2]{fcocquemas@fsu.edu}}
\author[id=au3,addressref={add3}]{\fnms{Ibrahim}~\snm{Ekren}\ead[label=e3]{iekren@fsu.edu}}
\author[id=au4,addressref={add4}]{\fnms{Abraham}~\snm{Lioui}\ead[label=e4]{abraham.lioui@edhec.edu}}
%%%%%%%%%%%%%%%%%%%%%%%%%%%%%%%%%%%%%%%%%%%%%%
%% Addresses                                %%
%%%%%%%%%%%%%%%%%%%%%%%%%%%%%%%%%%%%%%%%%%%%%%
\address[id=add1]{%
\orgdiv{Jones Graduate School of Business and School of Social Sciences},
\orgname{Rice University}}

\address[id=add2]{%
\orgdiv{College of Business},
\orgname{Florida State University}}

\address[id=add3]{%
\orgdiv{Department of Mathematics},
\orgname{Florida State University}}

\address[id=add4]{%
%\orgdiv{},
\orgname{EDHEC Business School}}
\end{aug}

%% Put support info here.  Reminder: do not thank the handling coeditor anonymously or by name
%\support{We thank four anonymous referees. The first author gratefully acknowledges financial support from the National Science Foundation through Grant XXX-0000000.}
%
\support{I. Ekren gratefully acknowledges financial support from the NSF Grant DMS-2007826.}

\begin{abstract}
We establish connections between optimal transport theory and the dynamic version of the Kyle model, including new characterizations of informed trading profits via conjugate duality and Monge-Kantorovich duality.  We use these connections to extend the model to multiple assets, general distributions, and risk-averse market makers.  With risk-averse market makers, liquidity is lower, assets exhibit short-term reversals, and risk premia depend on market maker inventories, which are mean reverting.  We illustrate the model by showing that implied volatilities predict stock returns when there is informed trading in stocks and options and market makers are risk averse.  
\end{abstract}

%\begin{keyword}
%\kwd{First keyword}We establish connections between optimal transport theory and the dynamic version of the Kyle model, including new characterizations of informed trading profits via conjugate duality and Monge-Kantorovich duality.  We use these connections to extend the model to multiple assets, general distributions, and risk-averse market makers.  With risk-averse market makers, liquidity is lower, assets exhibit short-term reversals, and risk premia depend on market maker inventories, which are mean reverting.  We illustrate the model by showing that implied volatilities predict stock returns when there is informed trading in stocks and options and market makers are risk averse.
%\kwd{second keyword}
%\end{keyword}

\end{frontmatter}

\section{Introduction}

The \citet{Kyle:1985:ContinuousAuctionsInsider} model has been a workhorse model for understanding the role of asymmetric information and liquidity in financial markets.  The dynamic version of the model reflects the reality that, in most markets, large investors split their orders into small pieces to minimize price impacts.  The continuous-time version is especially tractable.
In this paper, we make a substantial extension of the continuous-time model---to multiple assets, to assets with general distributions, and to risk-averse market makers---by establishing a connection with optimal transport theory.  Optimal transport theory has been used in economics in connection with matching problems and other topics \citep{Galichon:2016:OptimalTransportMethods}, but it has had only limited applications to asymmetric information in financial markets.  Our extension accommodates distributions that are not absolutely continuous, including discrete distributions, so it is an extension even for the single-asset/risk-neutral model.  By applying conjugate duality and Monge-Kantorovich duality, we also obtain characterizations of the gains from informed trading that are new even for the single-asset/risk-neutral model.

Our extension to risk-averse market makers is motivated by intermediary asset pricing theory \citep{He-Krishnamurthy:AER2013, He-Krishnamurthy:ARFE2018}.
We envision an  inter-dealer market in which dealers are price takers and in which there is a representative investor (dealer).   The representative dealer's marginal utility evaluated at aggregate dealer wealth serves as a stochastic discount factor in the inter-dealer market.   Competition between dealers ensures that the price they offer to non-dealer traders is the price at which they can trade in the inter-dealer market and hence is the price determined by the representative dealer's marginal utility.

By including risk-averse market makers in the Kyle model, we merge the two main theories of market liquidity: adverse selection and dealer aversion to inventory risk.  We can quantify the contribution each makes to market illiquidity.  
With either risk-neutral or risk-averse market makers, price changes are driven by orders, and the stochastic matrix (Kyle's lambda) that relates  orders to price changes is always symmetric and positive semidefinite.  We show that market liquidity is lower (lambdas are larger in the partial order of positive definiteness) when market makers are more risk averse. This is consistent with models of the bid-ask spread based on dealer aversion to inventory risk \citep{stoll}.  Furthermore, risk aversion produces mean reversion in market maker inventories, whereas they are a random walk in the risk-neutral model.  The mean reversion in inventories produces expected price changes (risk premia)  via the lambda matrix.

Our model exhibits excess volatility  and short-term return reversals, as in other models of inventory risk \citep{campbell-kyle,jegadeesh-titman}.
 When the informed trader and noise traders have been selling in aggregate, prices fall due to both the informational component of orders and due to the rising risk premia that result from rising dealer inventories.  The reverse is true when the informed trader and noise traders have been buying.  Thus, price changes exceed the changes that are induced by information alone.  The `excess' price changes are reversed on average as risk diminishes over time and risk premia are realized.  The precise expression of the excess volatility phenomenon that we derive is that quadratic variations (cumulative minute-to-minute variation) of price processes  exceed unconditional variances (long-run variation) of terminal prices.  Quadratic variations are higher because they equal risk-neutral variances, which are higher than variances under the real probability measure of the economy because the risk-neutral distribution puts high weights on `bad' states. For market makers, `bad' states are states in which the informed trader has extreme information, so the tails get more weight under the risk-neutral distribution.

As an application, we study informed trading in an underlying asset and an option on the asset.   There is evidence that informed traders do sometimes use options markets \citep{cao-chen-griffin, armental, hu}.  We show that when an informed trader can trade in options and market makers are risk averse, option-implied volatilities should predict the return of the underlying asset: future returns are higher in our model when implied volatilities are higher.  This is consistent with empirical evidence presented by \citet{anetal}.  This type of predictability cannot arise with risk-neutral market makers, because expected returns always equal the risk-free rate when market makers are risk neutral.  We also show that, regardless of whether noise trades in an option and underlying asset are positively or negatively correlated, risk-averse market makers usually end up with hedged positions---for example, long the underlying asset if they are short a call.  This is again contrary to the risk-neutral model.

\section{Literature Review}

 \citet{Cetin:2016:MarkovianNashEquilibrium} and \citet{Ying:2020WP} introduce risk-averse market makers into the dynamic Kyle model.  The difference between Cetin and Danilova's model and our model is roughly the difference between Bertrand and Walrasian competition.  Cetin and Danilova assume that a representative dealer always fulfills the net demand of the informed and noise traders at a price such that the dealer is indifferent about trading.  On the other hand, we assume that the dealer fulfills the demand at a price such that, after fulfilling it, the dealer is indifferent at the margin about making any further trades at that price.  As stated above, our assumption is motivated by the idea that dealers are price takers in an inter-dealer market.  Our assumption is the same as the assumption made in the intermediary asset pricing literature referenced above.   \citet{Ying:2020WP} takes a standard representative investor approach to pricing.  His representative investor consumes the aggregate dividend of the economy.  Apparently, gains and losses from trading with informed and noise traders are shared broadly across investors and have only a negligible effect on the market's pricing kernel.  Instead of evaluating the representative dealer's marginal utility at the economy's aggregate dividend as Ying does, we evaluate the representative dealer's marginal utility at aggregate dealer profits/wealth.  Again, our approach is consistent with the intermediary asset pricing literature in that it is dealer wealth rather than the aggregate dividend that determines asset prices.

Our premise that dealers behave competitively in an inter-dealer market has roots in the intermediary asset pricing literature, and it is also the premise of  other market microstructure studies, including \citet*{Naik:RFS1999} and \citet*{lester_etal}.  They consider models with only a single risky asset and investigate issues that differ from those we study. 

The topics we look at have been analyzed to some extent in single-period models.     \citet*{bollen_etal} combine inventory risk and adverse selection in a model of bid-ask spreads, modeling the inventory cost as the price of an option on an underlying asset the value of which depends on whether the dealer is trading with an informed or with an uninformed trader.  \citet{Subra1991} analyzes risk aversion on the part of the informed trader and market makers in a single-period Kyle model with a normally distributed asset.  Gaussian risk-neutral multi-asset single-period Kyle models are studied by 
 \citet{Caballe:1994:ImperfectCompetitionMultiSecurity}, \citet{Pasquariello:2015:StrategicCrossTradingStock} and \citet{GarciadelMolino:2020:MultivariateKyleModel}.  Single-period models of informed trading in options include \citet{Biais:1994:InsiderLiquidityTrading} and \citet{Easley:1998:OptionVolumeStock}.  The first of these considers an underlying asset that can take only three possible values and assumes, as in the Kyle model, that competitive market makers set prices after seeing an order from either an informed or uninformed trader.  The second considers an underlying asset that can take only two possible values and assumes competitive market makers quote bid and ask prices at which informed and uninformed traders can trade.
The only application of optimal transport theory to informed trading models of which we are aware is 
 \citet{Kramkov:2019:OptimalTransportProblem}.  They study a version of the  \citet{Rochet:1994:InsiderTradingNormality} model, which is a variant of the single-period Kyle model in which the informed trader can condition her order on realized noise trades.  
 
The continuous-time Kyle model with risk-neutral market makers has been applied and extended many times. A short list would include \citet{Back:1992:InsiderTradingContinuous}, \citet*{Back:2000:ImperfectCompetitionInformed}, \citet{Baruch:2002:InsiderTradingRisk}, \citet{Back:ECMA2004}, \citet{Caldentey:2010:InsiderTradingRandom}, \citet*{CCD2013}, \citet{Anderson-Smith:AER2013}, \citet{Collin-Dufresne:2016:InsiderTradingStochastic},  \citet{Cetin:2018:FinancialEquilibriumAsymmetric}, and \citet{Back:2018:ActivismStrategicTrading}.  Especially relevant to our work are the papers that study multiple assets. 
 \citet{Lasserre:2004:AsymmetricInformationImperfect} analyzes multiple assets in what is essentially a Gaussian model.  He assumes the vector of asset values is homeomorphic to a Gaussian vector, so market makers can filter for the Gaussian vector and then compute the conditional distribution of the asset values via the homeomorphism.  This is a special case of our model that excludes the case of options on an underlying asset that we study.  \citet{Back:1993:AsymmetricInformationOptions} and \citet{Back:2015:InformationalRoleStock} also study multi-asset versions of the continuous-time Kyle model.  The first of these studies a normally distributed asset value and a call option on the asset and makes a special parametric assumption to establish the existence of an equilibrium.  The second studies informed trading in a stock and bond but assumes that the informed trader's signal has only two possible values.  Our results include both of these as special cases.
%%%%%%%%%%%%%%%%%%%%%%%%%%%%%%%%%%%%%%%%%%%%%%%%%%%%%%%%%%%%%
%
\section{Optimal Transport with Risk Neutrality}\label{s:riskneutral}
%
%%%%%%%%%%%%%%%%%%%%%%%%%%%%%%%%%%%%%%%%%%%%%%%%%%%%%%%%%%%%%

Here, we describe the model with risk-neutral market makers and its equilibrium. We extend prior literature by allowing for multiple assets and distributions that are not absolutely continuous, including distributions supported on lower-dimensional spaces, as occurs with derivative securities.  The only assumption we make regarding the distribution of asset values is that the covariance matrix is finite.  We also give new characterizations of the expected profit of the informed trader.

There is a risk-free asset with risk-free rate normalized to zero.
There are $n$ risky assets.  The vector $\tv$ of risky asset values satisfies $\mye\left[\|\tv\|^{2}\right] < \infty$.  Let $F$ denote the distribution function of $\tv$.  The assets are traded on the time horizon $[0,T]$.  At date~$T$, the vector $\tv$ is publicly revealed.
There is a single informed trader who observes $\tv$ at date 0 and is risk neutral.  There are also noise (or ``liquidity'') traders whose cumulative trades form a vector Brownian motion $Z$ with zero drift and instantaneous covariance matrix $\Sigma$.\footnote{We can extend our results to noise trades satisfying $\D Z_t = \sigma(t)\,\D W_t$ where $W$ is a vector of independent Brownian motions and $\sigma$ is a continuous bounded matrix-valued function of time that is nonsingular for all $t$ and is such that $\sigma^{-1}$ is also a bounded function of $t$.  However, to economize on notation, we take $\sigma$ to be constant.  This produces the constant instantaneous covariance matrix $\Sigma = \sigma \sigma'$.}  Let $G$ denote the distribution function of $Z_T$, which is normal $(0,T\Sigma)$.  Let $X_t$ denote the vector of positions of the informed trader in the risky assets at date $t$.  Due to the risk neutrality of the informed trader, it is without loss of generality to take $X_0=0$.  We are going to allow the informed trader to use a mixed strategy.  We denote a random vector used for mixing by $\tu$.  It takes values in $\R^n$ and has an absolutely continuous distribution.  We require $X$ to be a continuous semimartingale relative to the filtration generated by $\tu$, $\tv$, and $Z$.\footnote{It is not necessary to assume that the informed trader observes the noise trades directly, because in equilibrium the informed trader can infer them by observing prices.  Also, we could allow jumps (discrete trades) in $X$, but it is  straightforward to show that jumps in $X$ are suboptimal, as in \citet{Back:1992:InsiderTradingContinuous}, so we exclude them for the sake of brevity.}  Set $Y=X+Z$.  The differential $\D Y_t$ of $Y$ at each date~$t$ is the net market order at~$t$.

Risk-neutral market makers observe the net order process $Y$ and compete to fill the orders.  Competition and risk neutrality force prices to equal expected values conditional on the information in orders, which we write as 
\begin{equation}\label{eq1}
    P_t = \mye[\tv \mid \cF^Y_t]\,.
\end{equation}
We look for an equilibrium in which $Y_t$ is a sufficient statistic for $P_t$ and denote the pricing rule by $P_t = H(t,Y_t)$.  In addition to \eqref{eq1}, the other equilibrium condition is that the informed trader's strategy is optimal, given the pricing rule $H$.  The informed trader's realized profit is
\begin{equation}\label{profit1}
\int_0^T (\tv-P_t)'\,\D X_t - \langle P,X\rangle_T\,,
\end{equation}
where $\langle \cdot,\cdot   \rangle$  denotes the sharp bracket process of the continuous vector semimartingales $P$ and $X$.  This formulation of the informed trader's profit follows from the intertemporal budget constraint of \citet{Merton:JET1971} via integration by parts, as shown by \citet{Back:1992:InsiderTradingContinuous} in the univariate case.\footnote{The formula \eqref{profit1} includes any gain or loss at the announcement date $(\tv-P_T)'X_T$ resulting from a jump in prices upon announcement, though such jumps never occur in equilibrium, that is, $P_T=\tv$.  A formula for the sharp bracket is given in Equation \eqref{bidask}.}  In equilibrium, we will have $\D X_t = \theta_t\,\D t$ for a vector process $\theta$, and in this case the realized profit is
$
 \int_0^T (\tv-P_t)'\theta_t\,\D t$,
which is the formulation assumed by \citet{Kyle:1985:ContinuousAuctionsInsider} in the univariate case.  The equilibrium condition is that the informed trader maximizes the expected value of \eqref{profit1}, conditional on $\tv$ and recognizing that $P_t = H(t,Y_t)$ with $Y = X+Z$.

The following theorem from optimal transport theory is key to the construction and characterization of an equilibrium.
Our contribution to the following is to establish the last statement.  We use that statement to construct the equilibrium mixed strategy of the informed trader when $\tv$ does not have an absolutely continuous distribution.  All proofs are in the appendix.  

%%%%%%%%%%%%%%%%%%%%%%%%%%%%%%%%%%%%%%%%%%%%%%%%%
%
% Theorem 3.1 (Brenier)
%
%%%%%%%%%%%%%%%%%%%%%%%%%%%%%%%%%%%%%%%%%%%%%%%%%

\begin{theorem}[Corollary to Brenier's Theorem]\label{thm_brenier}
There exists a unique convex function $\Gamma:\R^n\to \R$ such that $\nabla \Gamma(Z_T)$ has distribution function $F$ and $\mye[\Gamma(Z_T)]=0$. If $F$ is absolutely continuous, then $\nabla \Gamma $ is invertible (on the support of $F$) and $\tzeta := (\nabla \Gamma)^{-1}(\tv)$ has distribution function $G$.  If $F$ is not absolutely continuous, there still exists a random vector $\tzeta$ depending on $\tu$ and $\tv$ such that $\nabla \Gamma(\tzeta)= \tv$ and $\tzeta$ has distribution function $G$.
\end{theorem}

The function $\Gamma$ in this result is called a Brenier potential.  The gradient $\nabla \Gamma$ is called a transport map.  It transports the distribution $G$ to the distribution $F$.  The transport map $\nabla \Gamma$ solves the Monge problem with quadratic objective \citep[][Theorem 6.5]{Galichon:2016:OptimalTransportMethods} and hence is an optimal transport. The potential $\Gamma$ is very useful for describing equilibrium trading profits, in part because of Monge-Kantorovich duality, as we explain below.  The condition $\mye[\Gamma(Z_T)]=0$ in Theorem~\ref{thm_brenier} is a normalization; absent this normalization, the potential is unique only up to an additive constant.  The normalization simplifies the description of equilibrium trading profits.

Let $k(t,y,z)$ denote the transition density of the noise trade process $Z$ from $(t,y)$ to $(T,z)$; that is, $k(t,y,\cdot)$ is the normal density function with mean vector $y$ and covariance matrix $(T-t)\Sigma$.  Consider the following pricing rule: Set $P_0 = \mye[\tv]$ and, for $t>0$ and $y \in \R^n$, set $P_t = H(t,Y_t)$ where
\begin{equation}\label{price1}
H(t,y) := \int_{\R^n}  \nabla \Gamma(z)k(t,y,z)\,\D z\,.
\end{equation}
With some abuse of notation, we set 
\begin{equation}\label{Gammaty}
\Gamma(t,y) := \int_{\R^n}  \Gamma(z)k(t,y,z)\,\D z \,.
\end{equation}
The normalization in Theorem~\ref{thm_brenier} means that $\Gamma(0,0)=0$.  
The following lemma states that $\Gamma$ is sufficiently smooth to apply It\^o's lemma and also that we can interchange differentiation and expectation on the right-hand side of \eqref{Gammaty} to compute the gradient in $y$ of $\Gamma(t,y)$, which we denote as $\nabla \Gamma (t,y)$.  

%%%%%%%%%%%%%%%%%%%%%%%%%%%%%%%%%%%%%%%%%%%%%%%%%
%
% Lemma 3.1 (dGamma)
%
%%%%%%%%%%%%%%%%%%%%%%%%%%%%%%%%%%%%%%%%%%%%%%%%%

\begin{lemma}\label{lem_dGamma}
$\Gamma(t,y)$ is continuously differentiable in $t$ and twice continuously differentiable in $y$, and $\nabla \Gamma(t,y) = H(t,y)$  for $t\in (0,T)$ and $y \in \R^n$.  %Furthermore,
%\begin{equation}\label{dGamma}
%P_t'\,\D Y_t =  \D \Gamma(t,Y_t) \,.
%\end{equation}
\end{lemma}

We can now prove the existence of equilibrium.  We need a mild restriction on trading strategies, stated as \eqref{restriction} below, which is described as a `no doubling strategies' condition in \cite{Back:1992:InsiderTradingContinuous}. It  holds, for example, if $\mye\int_0^T \|H(t,X_t+Z_t)\|^2\,\D t < \infty$.  The function $\Gamma^*$ in the following theorem is the convex conjugate (Fenchel transform) of the convex function $\Gamma$, defined as
$\Gamma^*(v) = \sup\,\{v'y - \Gamma(y) \mid y \in \R^n\}$.

%%%%%%%%%%%%%%%%%%%%%%%%%%%%%%%%%%%%%%%%%%%%%%%%%
%
% Theorem 3.2 (Risk-Neutral Equilibrium)
%
%%%%%%%%%%%%%%%%%%%%%%%%%%%%%%%%%%%%%%%%%%%%%%%%%

\begin{theorem}\label{thm:riskneutral}
Let $\tzeta$ be the random vector given in Theorem~\ref{thm_brenier}.
Given the pricing rule \eqref{price1}, the strategy $\D X_t = \theta_t\,\D t$ where
\begin{equation}\label{brownbridge}
\theta_t = \frac{1}{T-t}(\tzeta - Y_t)
\end{equation}
maximizes the informed trader's profit
in the class of continuous semimartingales $X$ such that
\begin{equation}\label{restriction}
\mye \int_0^T H(t,X_t+Z_t)'\,\D Z_t = 0\,.
\end{equation}
Furthermore, given the trading strategy \eqref{brownbridge}, the pricing rule \eqref{price1} satisfies the equilibrium condition \eqref{eq1}.  The maximum expected profit of the informed trader, as of date~0 and conditional on $\tv$, is $\Gamma^*(\tv)$, and, for each $(t,v,y)$,
\begin{align}
J(t,v,y) :&= \sup_X \; \mye \left[\left.\int_t^T (v-P_u)'\,\D X_u - \int_t^T \D \langle P,X \rangle_u\,\right|\, \tv=v, Y_t=y\right] \notag\\
&= \Gamma^*(v) + \Gamma(t,y) - y'v\,, \label{valuefunction}
\end{align}
where the supremum is taken over trading strategies $X$ satisfying \eqref{restriction}.  The price vector $P$ evolves as $\D P_t = \Lambda_t\,\D Y_t$, where $\Lambda_t = \nabla^2 \Gamma(t,Y_t)$, and the matrix $\Lambda_t$ is symmetric and positive semidefinite at each date~$t$ and in each state of the world. The net order process $Y$ is a $(0,\Sigma)$--Brownian motion relative to market makers' information.
\end{theorem}

To explain the construction of the equilibrium, and the sense in which it is unique, consider trading strategies of the form $\D X_t = \theta_t\,\D t$.  It is easy to calculate that the maximization problem in the Hamilton-Jacobi-Bellman (HJB) equation has no solution unless each element $H_i$ for $i=1,\ldots,n$ of the pricing rule $H$ satisfies the heat equation: 
\begin{align}\label{heateqn}
\frac{\partial H_i}{\partial t} + \frac{1}{2}\tr(\Sigma \nabla^2 H_i) = 0\,.
\end{align}
This argument is the same as in \citet{Back:1992:InsiderTradingContinuous} for the univariate model and is based on the linearity of the HJB equation in the control $\theta$.  Also, given the heat equation, it is possible to show (we do this in the proof of Theorem~\ref{thm:riskneutral}) that a trading strategy is optimal if and only it is of finite variation and pushes the price vector to $\tv$ at the end of trading, meaning $H(T,Y_T)=\tv$.  The multiplicity of optimal strategies is due to risk neutrality and the fact that the informed trader can continuously move up and down the inverse supply curve $H(t,\cdot )$ posted by market makers, like a perfectly discriminating monopolist/monopsonist---mathematically, this takes the form that the HJB maximization problem is solved by any $\theta$ when \eqref{heateqn} holds (and by no $\theta$ when it does not hold).  By the Feynman-Kac theorem, the heat equation \eqref{heateqn} is equivalent to 
\begin{equation}
    H_i(t,y) = \mye[H_i(T,Z_T) \mid Z_t=y]\,.
\end{equation}
On the other hand, the equilibrium condition \eqref{eq1} in conjunction with $H(T,Y_T)=\tv$ implies
\begin{equation}
    H_i(t,y) = \mye[H_i(t,Y_T) \mid Y_t=y]\,.
\end{equation}
This suggests that $Y$ and $Z$ must have the same distribution.  Now, we observe that the conditions $Y \sim Z$ and $H(T,Y_T)=\tv$ imply that $H(t,Z_T)$ must have the same distribution as $\tv$.  In other words, \textit{$H(T,\cdot)$ must transport $G$ to $F$.}  Thus, we are led to the definition $H(T,y) = \nabla \Gamma(y)$ and, due to the heat equation and Feynman-Kac theorem, to the pricing rule \eqref{price1}.\footnote{In the univariate case with a continuous strictly increasing distribution function $F$, the unique monotone map that transports $G$ to $F$ is $F^{-1}\circ G$; hence $\nabla \Gamma = F^{-1}\circ G$.  This produces the pricing rule in \citet{Back:1992:InsiderTradingContinuous}.}  

We show in the proof of Theorem~\ref{thm:riskneutral} that, for any finite-variation strategy satisfying the regularity condition \eqref{restriction}, the informed trader's expected profit is 
$\mye[\tv'Y_T - \Gamma(Y_T)]$.  Thus, a finite-variation strategy is optimal if it results in $Y_T \in \argmax \{\tv'y - \Gamma(y)\}$, which is equivalent to $\tv$ being in the subdifferential of $\Gamma(Y_T)$, which, on a set of full Lebesgue measure, is equivalent to $\nabla \Gamma(Y_T)=\tv$.  Given that $P_T = H(T,Y_T) = \nabla \Gamma(Y_T)$, this is equivalent to $P_T = \tv$.  Thus, as stated above, any finite-variation strategy that pushes $P_t$ to $\tv$ as $t \rightarrow T$ is optimal.  Furthermore, this shows that the maximum expected profit conditional on $\tv$ is $\Gamma^*(\tv)$ as stated in the theorem.  This result, and the more general formula \eqref{valuefunction} for the value function, are new results even for the univariate model.

The strategy \eqref{brownbridge} is the drift of a Brownian bridge ending at $\tzeta$.  A Brownian bridge is a Brownian motion conditioned on the ending value, so, conditional on $\tv$, $Y$ is a Brownian motion conditioned to end at $\tzeta$.  In other words, the informed trader knows in advance the ending point of $Y$, which she controls through the drift, and sees $Y$ as a Brownian bridge.  Because market makers do not observe $\tzeta$ and because $\tzeta$ has the same distribution as $Z_T$, market makers see $Y$ as a Brownian motion.  They attempt to forecast its ending value and thereby to forecast $H(T,\tzeta) = H(T,Y_T)=\tv$.  This induces the pricing rule \eqref{price1}.

In the univariate Gaussian model studied by \citet{Kyle:1985:ContinuousAuctionsInsider}, the price evolves as $\D P = \lambda \,\D Y$ for a constant $\lambda$ that is universally known as Kyle's lambda.  Theorem~\ref{thm:riskneutral} shows that price changes are also linearly related to orders in multivariate non-Gaussian models, with the Kyle lambda matrix being symmetric and positive semidefinite.  
Symmetry and positive semidefiniteness of the $\Lambda$ matrix is shown in a single-period Kyle model by \citet{Caballe:1994:ImperfectCompetitionMultiSecurity}.  In Theorem~\ref{thm:riskneutral}, symmetry follows from the matrix being the Hessian of the Brenier potential, and positive semidefiniteness follows from the convexity of the potential.  

We can further characterize the unconditional expected profit of the informed trader via Monge-Kantorovich duality.  We have
\begin{equation}\label{mk}
\sup_{\Phi\in \Pi(F,G)} \int v'y \,\Phi(\D v,\D y) = \mye[\Gamma^*(\tv)] = \inf_{f,g} \; \mye[f(\tv)] + \mye[g(Y_T)]\,,
\end{equation}
where $\Pi(F,G)$ is the set of probability measures $\Phi$ on $\R^n \times \R^n$ for which the marginal distributions are the exogenously given $F$ and $G$, where the infimum is taken over all functions $f$ and $g$ with the property that $f(x) + g(y) \ge x'y$ for all $(x,y)$, and where the expectations are taken over the distributions $F$ for $\tv$ and $G$ for $Y_T$.  For this result, see Sections~2.3 and~6.2 of \citet{Galichon:2016:OptimalTransportMethods}.
The first equality in \eqref{mk} provides the interpretation that it is as if the informed trader earns $v'y$ and can choose any joint distribution for $v$ and $y$ with the given marginals.  Such a joint distribution is called a coupling.  A map $v \mapsto y$ in concert with the distribution of $\tv$ determines a coupling, provided the induced distribution of $y$ is~$G$.  Thus, it is as if the informed trader chooses $Y_T$ as a function of $\tv$ to maximize $\mye[\tv'Y_T]$ subject to the constraint that $Y_T$ has the same distribution as $Z_T$.\footnote{The second equality in \eqref{mk} provides another characterization of the informed trader's expected profit, but it seems somewhat less meaningful.  The equality shows that the expected profit is what the informed trader would achieve if her profit were additively separable as $f(v)+g(y)$ with the given distributions $F$ and $G$ for $v$ and $y$, and if market makers could choose $f$ and $g$ to minimize the  expected profits, with the proviso that the profit $f(v)+g(y)$ could never be smaller than $v'y$.  It turns out that the functions $f$ and $g$ that achieve this minimum are $\Gamma^*(v)$ and $\Gamma(y)$ \citep[][Proposition 6.4]{Galichon:2016:OptimalTransportMethods}.}

The expected profit of the informed trader is related to  the Wasserstein-2 distance between $F$ and $G$. Indeed, the unconditional expected profit \eqref{mk} can also be written as
$$\frac{\| \tv \|^2 + \| Z_T \|^2 - W^2_2(F,G)}{2},$$
where $\| \cdot \|$ denotes the $\mathcal{L}^2$ norm and where
$$W^2_2(F,G) := \inf_{\Phi\in \Pi(F,G)} \int |v-y|^2\Phi(dv,dy)$$
is the square of the Wasserstein-2 distance between $F$ and $G$.
Thus, the expected profit depends on the amount of noise trading as measured by $\| Z_T \|^2$ and, given the distribution $G$ of noise trading, depends on the amount of private information as measured by $\| \tv \|^2 - W_2^2(F,G)$.  

To set the stage for the next section, we compute the aggregate dealer profits in the equilibrium of Theorem~\ref{thm:riskneutral}.  Aggregate dealer profits equal the negative of the sum of the informed and noise traders' profits.  The profits of noise traders are  given in \eqref{profit1}, replacing $X$ with $Z$.  Thus, realized dealer profits are
\begin{equation}\label{dealerprofits}
\tw := \langle P,Y \rangle_T - \int_0^T (\tv - P_t)'\,\D Y_T\,.
\end{equation}
This formula can also be derived directly from Merton's intertemporal budget constraint, viewing the dealers as investors with position $-Y_t$ in the risky asset.
Extending our earlier discussion of expected informed trader profits in terms of the Brenier potential $\Gamma$ and its conjugate, we can show the following.

%%%%%%%%%%%%%%%%%%%%%%%%%%%%%%%%%%%%%%%%%%%%%%
%
% Corollary 3.1
%
%%%%%%%%%%%%%%%%%%%%%%%%%%%%%%%%%%%%%%%%%%%%%%

\begin{corollary}\label{cor21} 
Given the pricing rule \eqref{price1}, if the informed trader follows a finite-variation strategy, then dealer profits \eqref{dealerprofits} equal
$\tw = \langle P,Y \rangle_T + \Gamma(Y_T) - \tv'Y_T$.
When the informed trader follows an optimal strategy, dealer profits are
$\tw =  \langle P,Y \rangle_T - \Gamma^*(\tv)$.
\end{corollary}

The sharp bracket quantity in \eqref{dealerprofits} and in the corollary is
\begin{equation}\label{bidask}
\langle P,Y\rangle_T = \sum_{i=1}^n \sum_{j=1}^n \sigma_{ij}\int_0^T \frac{\partial^2 \Gamma(t,Y_t)}{\partial y_i \partial y_j}\,\D t = \text{tr}\left(\Sigma\int_0^T \nabla^2\Gamma(t,Y_t)\,\D t\right)\,,
\end{equation}
where $\sigma_{ij}$ is the $(i,j)$th element of the instantaneous covariance matrix $\Sigma$ of $Z$.   This is the `bid-ask spread' cost that noise traders pay market makers, which, on average, offsets the market makers' losses $\Gamma^*(\tv)$ to the informed trader.
It is noteworthy that dealer profits depend on the path of noise trading only via the bid-ask spread costs.  If noise traders make bad trades, then the informed trader will make more money by subsequently reversing them, but there is no net gain or loss to market makers.  Noise traders are as likely to make good trades as bad trades, so on average this is also a wash for the informed trader, whose expected profit is simply $\mye[\Gamma^*(\tv)]$.

%%%%%%%%%%%%%%%%%%%%%%%%%%%%%%%%%%%%%%%%%%%%%%
%
\section{Risk Aversion}\label{s:riskaverse}
%
%%%%%%%%%%%%%%%%%%%%%%%%%%%%%%%%%%%%%%%%%%%%%%

We now assume there is a representative dealer with CARA utility.  Assets are priced by the representative dealer's marginal utility evaluated at aggregate dealer wealth. Our strategy is to work under the risk-neutral probability. 
It is common in derivative security pricing to take the risk-neutral distribution of the underlying asset as given and then to derive the values of derivative securities.  
%This is true also in term structure modeling, in which the risk-neutral distribution of the short-term interest rate is commonly taken as given and bond prices are derived.  
We are going to follow that approach and take the risk-neutral distribution of $\tv$ as given.  
The risk-neutral distribution will be consistent with a unique actual probability measure for the economy, which for brevity we call the physical probability.  Thus, we will set up a map from risk-neutral distributions of $\tv$ to physical distributions of $\tv$.  Our approach proves the existence of equilibrium for all physical distributions in the range of the map.  We show in Section~\ref{s:normal} that the range of the map includes all normal distributions, and we analytically invert the map to obtain the equilibrium quantities in terms of the physical distribution whenever the risk-neutral distribution is normal.  We show how to compute the physical distribution numerically in other cases.

Due to risk aversion, any initial inventory held by dealers will affect pricing.  We let $\beta$ denote the number of shares held by market makers at date~0.  We continue to set $Y=X+Z$, where $X$ is the number of shares purchased by the informed trader and $Z$ is the number of shares purchased by noise traders.  This implies that the number of shares held by market makers at any date~$t$ is $\beta-Y_t$.
The equilibrium condition for prices can be expressed as
\begin{equation}\label{eq2}
P_t = \mye^\Q[\tv \mid \mathcal{F}^Y_t]\,,
\end{equation}
where $\Q$ denotes the risk-neutral probability and $\mathcal{F}^Y_t$ denotes market makers' information (the history of $Y$ prior to $t$).
We continue to assume that the informed trader is risk neutral, so she maximizes expected profits under the physical probability.  The link between the risk-neutral and physical probabilities is that
\begin{equation}\label{dQdP}
\frac{\D \Q}{\D \P} = \frac{\E^{-\alpha \tw}}{\mye[\E^{-\alpha \tw}]}\end{equation} 
where $\P$ denotes the physical probability, $\alpha$ denotes the absolute risk aversion of the representative dealer, and $\tw$ denotes aggregate dealer wealth at date~$T$.  The expression \eqref{dQdP} is the stochastic discount factor (SDF).  Define the SDF process
$M_t = \mye\left[\D \Q/\D \P \mid \cF_t^Y\right]$.

It may be useful to describe our basic approach here before getting into details.  In the risk-neutral case, market orders $\D Y$ always have a zero mean given market makers' information.  In other words, $Y$ is a $(\P,\cF^Y)$-martingale.  In the risk-averse case, $Y$ will be a $(\Q,\cF^Y)$-martingale.  Equilibrium prices will be as in the previous section, but under the risk-neutral probability; that is,  $P_t = \nabla \Gamma(t,Y_t)$ where $\Gamma(y)$ is the Brenier potential for transporting the distribution of $Z_T$ to the risk-neutral distribution of $\tv$ and where $\Gamma(t,y)$ is defined from the potential as in the previous section.  Kyle's lambda matrix also has the same form as in the previous section: it is the Hessian of $\Gamma(t,y)$.   Girsanov's theorem, martingale representation, and some related results allow us to construct a function $\gamma$ and a  $(\P,\cF^Y)$-Brownian motion $\hat{Y}$ so that $\D M_t/M_t = - \gamma(t,Y_t)'\,\D \hat Y_t$.
The vector $\gamma(t,Y_t)$ is the vector of `prices of risk.'    It follows from standard asset pricing theory that risk premia equal minus the covariances of returns with the SDF process, so the drifts of prices under the physical probability, relative to market makers' information, are $-\D \langle M,P\rangle /M$. 
Using the formula for Kyle's lambda matrix, we compute that the drifts are
$\nabla^2\Gamma(t,Y_t)\Sigma\gamma(t,Y_t)\,\D t$.
Thus, we obtain the equilibrium risk premia when market makers are risk averse.

For clarity, we  assume, without loss of generality, that the underlying probability space is a product space $\Omega := \R^n \times \R^n \times C[0,T]^n$, with generic element $(u,v,z)$, where $u \in \R^n$ denotes the variable used for mixing, $v \in \R^n$ denotes the vector of asset values, and $z \in C[0,T]^n$ denotes the path of cumulative noise trades.  We continue to let $G$ denote the normal $(0,T\Sigma)$ distribution. Let $\nu$ denote the $(0,\Sigma)$--Wiener measure on $C[0,T]^n$---i.e., the distribution of a Brownian motion with zero drift and instantaneous covariance matrix $\Sigma$. We will construct $\P$, the physical probability measure on $\Omega$, as the product of three measures: the distribution $G$ for the mixing variable $\tu$, a distribution $\hat F$ for the vector $\tv$ of asset values, and the Wiener measure $\nu$ on $C[0,T]^n$, so the mixing variable, asset values, and noise trades will all be independent under $\P$.  The random vectors $\tu$ and $\tv$ and the vector Brownian motion $Z$ are defined on this product space as the projection maps: $\tu(u,v,z) = u$, $\tv(u,v,z)=v$, and $Z(u,v,z)=z$. We will indicate the measures under which expectations are taken with superscripts; for example, $\mye^\nu$ denotes expectation with respect to the Wiener measure $\nu$. 
Because $G$ and $\nu$ are fixed,  we need to specify $\hat F$, the $\P$-distribution of $\tv$, in order to determine $\P$.

Let ${F}$ be a distribution function on $\R^n$.  Under some assumptions, we will construct a physical distribution $\hat F$ for $\tv$ so that there exists an equilibrium with $F$ being the risk neutral distribution of $\tv$.  
Let $\Gamma$ denote the Brenier potential such that $\nabla \Gamma$ transports $G$ to ${F}$.
Define a pricing rule $H(t,y)$ as in \eqref{price1} and define $\Gamma(t,y)$ from $\Gamma$ as in \eqref{Gammaty}.    As shown in  Lemma~\ref{lem_dGamma}, we have
$H(t,y) = \nabla \Gamma(t,y)$.
We choose the arbitrary additive constant in the potential so that $\Gamma(0,0)=0$ as in Section~\ref{s:riskneutral}, which is equivalent to the mean of $\Gamma(Z_T)$ being zero when $Z_T$ is normal $(0,T\Sigma)$.  The  informed trader's optimization problem is the same as in the risk-neutral model, so the informed trader's expected profit conditional on $\tv$ is $\Gamma^*(\tv)$ as in Theorem~\ref{thm:riskneutral}, and aggregate dealer profits are as stated in Corollary~\ref{cor21}.  We repeat the formula from Corollary~\ref{cor21} here but including the value of the initial position of~$\beta$ shares:
\begin{equation}\label{wealth3}
    \tw = \beta' \tv - \Gamma^*(\tv) + \langle P,Y \rangle_T\,.
\end{equation}

To identify $\hat F$ and the equilibrium informed trading strategy, we ask that $\hat F$ and the distribution for $Y$ on $(\P,\cF^Y)$ be consistent with (i) $Y$ is a $(0,\Sigma)$ Brownian motion on $(\Q,\cF^Y)$, (ii) $\D \P/\D \Q$ is as specified in \eqref{dQdP}, (iii) market makers' terminal wealth is as stated in \eqref{wealth3}, and (iv) $\nabla \Gamma(Y_T) = \tv$.  To find $\hat F$ and the distribution for $Y$ on $(\P,\cF^Y)$, we need to make some calculations involving a $(0,\Sigma)$ Brownian motion.  For this purpose, we use the $(0,\Sigma)$ Brownian motion $Z$ under $\nu$.  Any other $(0,\Sigma)$ Brownian motion would serve as well for this purpose.  
It follows from \eqref{dQdP} that $\D \P/\D \Q  = \E^{\alpha \tw}/\mye^{\Q}[\E^{\alpha \tw}]$.
Motivated by this formula, the formula $\nabla \Gamma(Y_T) = \tv$, and the formula \eqref{wealth3} for dealers' terminal wealth, define 
\begin{align}\label{eq:deftx}
    \tilde\chi  &= \exp\left\{-\alpha(Z_T-\beta)'\nabla\Gamma(Z_T)+\alpha \Gamma(Z_T) + \alpha\tr\left(\Sigma \int_0^T \nabla^2\Gamma(t,Z_t)\,\D t\right)\right\}\,.
\end{align}
We need the following regularity condition, which will ensure sufficient integrability of the representative dealer's marginal utility and its reciprocal.  We later deduce this from conditions on primitives (Theorem~\ref{thm:riskaverse2}).

%%%%%%%%%%%%%%%%%%%%%%%%%%%%%%%%%%%%%%%%%%%%%%%%%
%
% Assumption 4.1 
%
%%%%%%%%%%%%%%%%%%%%%%%%%%%%%%%%%%%%%%%%%%%%%%%%%

\begin{assumption}\label{integrability}
$\tilde \chi^2$ and $\tilde \chi^{-1}$ have finite $\nu$ means, 
and the function $\phi:[0,T]\times \R^n\mapsto \R$ defined by
\begin{multline}
 \label{defphi}
\E^{\phi(t,z)} = \mye^\nu\bigg[\exp\bigg\{ -\alpha(Z_T-\beta)'\nabla\Gamma(Z_T)+\alpha \Gamma(Z_T) \\ \left. \left. \left. +\alpha \tr\left(\Sigma \int_t^T \nabla^2 \Gamma(u,Z_u)\,\D u  \right)\right\}\,\right|\, Z_t=z\right]\,
\end{multline}
is finite and continuous on $[0,T)\times \R^n.$
\end{assumption}

Define $\tilde \xi = \tilde \chi / \mye^\nu[\tilde \chi]$ and set $\xi_t =\mye^\nu [\tilde\xi \mid \cF^Z_t]$ for $t \in [0,T]$.  Define a change of measure $\D\hat\nu/\D \nu = \tilde \xi$. We fix $\hat F$ to be the $\hat\nu$--distribution of $\nabla \Gamma (Z_T)$. 
The equilibrium distribution of $Y$ on $(\P,\cF^Y)$ will be the $\hat\nu$--distribution of $Z$.  This is described in the following lemma.

%%%%%%%%%%%%%%%%%%%%%%%%%%%%%%%%%%%%%%%%%%%%%%%%%
%
% Lemma 4.1
%
%%%%%%%%%%%%%%%%%%%%%%%%%%%%%%%%%%%%%%%%%%%%%%%%%

\begin{lemma}\label{gmarkov}
Under Assumption 3.1, there exists a measurable function $\gamma:(0,T)\times \R^n\mapsto\R^n$ such that $\D \xi_t/\xi_t = \gamma(t,Z_t)'\,\D Z_t$ and $\D Z_t = \Sigma \gamma(t,Z_t)\,\D t + \D \hat Z_t$, where $\hat Z$ is a $(0,\Sigma)$--Brownian motion on $(\hat \nu, \cF^Z)$.
\end{lemma}

When $\nabla \Gamma$ is not invertible, we will allow the informed trader to mix as in the previous section.  The following establishes the existence of the desired mixing variable. 

%%%%%%%%%%%%%%%%%%%%%%%%%%%%%%%%%%%%%%%%%%%%%%%%%
%
% Lemma 4.2 (Mixing)
%
%%%%%%%%%%%%%%%%%%%%%%%%%%%%%%%%%%%%%%%%%%%%%%%%%

\begin{lemma}\label{lem_mix}
There exists a function $f : \R^{n} \times \R^n \rightarrow \R^n$ such that, setting $\tzeta = f(\tu,\tv)$, where the distribution of $(\tu,\tv)$ is $G \otimes \hat F$, we have (i) $\nabla \Gamma (\tzeta) = \tv$ almost surely, and (ii) the distribution of $\tzeta$ is the $\hat{\nu}$--distribution of $Z_T$.
\end{lemma}

Lemma~\ref{gmarkov} implies that $Z$ is a Markov process under $\hat \nu$. We denote  the $\hat \nu$ transition density of $Z$ from $(t,z)$ to $(T,b)$ by $h(t,z,b)$. Then, $h(t,Z_t,b)$ is a $\hat \nu$-martingale that can be represented as a $\hat Z_t$ stochastic integral.  We make the following regularity assumption for $h$, which we derive from conditions on primitives in Theorem~\ref{thm:riskaverse2}.

%%%%%%%%%%%%%%%%%%%%%%%%%%%%%%%%%%%%%%%%%%%%%%%%%
%
% Assumption 4.2
%
%%%%%%%%%%%%%%%%%%%%%%%%%%%%%%%%%%%%%%%%%%%%%%%%%

\begin{assumption}\label{assumeh0}
For all $t<T$, $h$ is differentiable in $z$, $\int | \nabla_z h(t,z,b)|\,\D b < \infty$, and
\begin{align}\label{assumeh}
\D h(t,Z_t,b)=\nabla_z h(t,Z_t,b)'(\D Z_t-\Sigma \gamma (t,Z_t)\,\D t)\,.
\end{align}
\end{assumption}

We remark that equation \eqref{assumeh} follows from  $h(t,Z_t,b)$ being a $\hat \nu$--martingale and It\^o's formula when $h$ is $C^{1,2}$ in $(t,z)$.
\begin{comment}
Now, we provide guidelines for a Doob $h$-transform of $\hat{Z}$. Proofs of these statements will be provided in the proof of Theorem \ref{thm:riskaverse}. For any $b$,
$$\D Z^*_t := \D \hat{Z}_t - \Sigma \nabla_z \log h(t,Z_t,b) \,\D t$$
is a Brownian motion relative to $\hat{\nu}$ conditioned on $Z_T=b$.
Thus, conditioned on the event $Z_T=b$, $Z$ has the dynamics
\begin{equation}\label{neweq}
    \D Z_t = \Sigma\gamma(t,Z_t)\,\D t + \Sigma \nabla_z \log h(t,Z_t,b) \,\D t + \D Z^*_t.
\end{equation}
It follows that, if $\tilde\zeta$ is a random vector whose distribution is the $\hat{\nu}$--distribution of $Z_T$ and independent of $Z^*$, then the solution $Z$ of \eqref{neweq} with $a=\tilde\zeta$ satisfies $Z_T = \tilde\zeta$ and has distribution $\hat\nu$ on its own filtration.  These are the properties we want for the aggregate order process $Y$, so the equilibrium informed trading strategy will be
\end{comment}
Our candidate for an equilibrium informed trading strategy is
\begin{equation}\label{newtheta}
    \theta_t = \Sigma\gamma(t,Y_t) + \Sigma \nabla_y \log h(t,Y_t,\tilde\zeta)\,.
\end{equation}
Given this strategy, the aggregate orders are\footnote{This defines $Y$ as a Doob $h$-transform of $Z$.  See, for example, \citet[][IV.39-40]{rogers2000diffusions}.  The process $Y$ defined by
$$\D Y_t = \Sigma\gamma(t,Y_t)\,\D t + \Sigma \nabla_y \log h(t,Y_t,b) \,\D t + \D Z_t$$
has on its own filtration the $\hat{\nu}$ distribution conditioned on ending at $b$.  Because $\tzeta$ has the terminal $\hat \nu$ distribution and is independent of $Z$, the process \eqref{eqy} has the $\hat \nu$ distribution on its own filtration and satisfies $Y_T = \tzeta$. We provide a proof in the appendix.}
\begin{equation}\label{eqy}
    \D Y_t = \Sigma\gamma(t,Y_t)\,\D t + \Sigma \nabla_y \log h(t,Y_t,\tilde \zeta) \,\D t + \D Z_t\,.
\end{equation}
We now show that our construction defines an equilibrium, under Assumptions~\ref{integrability} and~\ref{assumeh0} and assuming a unique strong solution to \eqref{eqy}.  We later provide conditions on primitives that guarantee these regularity conditions hold.

%%%%%%%%%%%%%%%%%%%%%%%%%%%%%%%%%%%%%%%%%%%%%%%%%
%
% Theorem 4.1 (Risk-Averse Equilibrium)
%
%%%%%%%%%%%%%%%%%%%%%%%%%%%%%%%%%%%%%%%%%%%%%%%%%

\begin{theorem}\label{thm:riskaverse}
  Let $F$ be a given distribution on $\R^n$ such that $\int \| v\|^2 \D F(v) < \infty$ and let $\Gamma(y)$ be the Brenier potential for transporting $G$ to ${F}$. Define $\Gamma(t,y)$ by \eqref{Gammaty}.  Assume Assumption~\ref{integrability} holds, and define $\tilde \xi$ by \eqref{eq:deftx} and $\hat\nu$ by $\D \hat \nu/\D \nu=\tilde \xi$.  Assume the physical distribution of $\tv$ is  $\hat F$ (the $\hat{\nu}$--distribution of $\nabla \Gamma (Z_T)$).  Then, the mean of $\tv$ is finite under the physical distribution.  Let $h$ denote the $\hat \nu$ transition density of $Z$ and assume  Assumption~\ref{assumeh0} holds.  Assume the stochastic differential equation \eqref{eqy} admits a unique strong solution.
Given the pricing rule $H(t,y) = \nabla \Gamma(t,y)$ and informed trading strategy \eqref{newtheta}, aggregate dealer wealth at date $T$ is given by \eqref{wealth3}.  Define the risk-neutral probability $\Q$ by \eqref{dQdP}.  Then, the risk-neutral distribution of $\tv$ is ${F}$, the equilibrium pricing condition \eqref{eq2} holds, and the informed 
trading strategy \eqref{newtheta} is optimal in the class of strategies satisfying \eqref{restriction}.   Furthermore, $Y$ is a $(0,\Sigma)$--Brownian motion on $(\Q,\cF^Y)$ and satisfies
$\D Y_t = \Sigma \gamma(t,Y_t)\,\D t + \D \hat Y_t$,
where $\hat Y$ is a $(0,\Sigma)$--Brownian motion on $(\P,\cF^Y)$.
\end{theorem}

We can obtain more explicit formulas for the vector of prices of risk and the equilibrium informed trading strategy.
The definition of $\xi$ and the definition \eqref{defphi} of $\phi$ imply that
\begin{align}\label{eq:phirho}
\xi_t = \exp\left\{\phi(t,Z_t) -\phi(0,0)+ \alpha\tr\left(\Sigma \int_0^t \nabla^2\Gamma(s,Z_s)\,\D s\right)\right\}\,.
\end{align}
If $\phi$ is smooth enough to apply It\^o's formula, then this formula and the fact that $\xi$ is a $\nu$--martingale imply
$\D \xi_t/\xi_t = \nabla \phi(t,Z_t)'\,\D Z_t$
and consequently 
\begin{align}\label{eq:gp}
\gamma(t,y) = \nabla \phi(t,y)\,.
\end{align}
The martingale property also implies the following partial differential equation (PDE) for $\phi$:
\begin{equation}\label{pde}
\frac{\partial \phi}{\partial t} + \frac{1}{2}\tr(\Sigma \nabla^2 \phi) + \frac{1}{2}\nabla \phi ' \Sigma \nabla \phi + \alpha \tr (\Sigma \nabla^2\Gamma) = 0\,
\end{equation}
with boundary condition $\phi(T,z)=-\alpha(z-\beta)'\nabla\Gamma(z)+\alpha \Gamma(z)$.  We can use this to calculate $\nabla \phi$.  Furthermore, 
Assumption~\ref{integrability} implies that we can define a function $\psi(t,z,b)$ by
\begin{multline}\label{defpsi}
\E^{\psi(t,z,b)} = \E^{-\alpha(b-\beta)'\nabla\Gamma(b)+\alpha \Gamma(b)} \\ \times \mye^\nu\left[\left.\exp\left\{ \alpha \tr\left(\Sigma \int_t^T \nabla^2 \Gamma(u,Z_u)\,\D u  \right)\right\}\,\right|\, Z_t=z, Z_T=b\right]\, .  
\end{multline}
Using the definition of $\hat \nu$, we have
\begin{align}\label{eq:h}
    h(t,z,b) = k(t,z,b)\E^{-\phi(t,z) + \psi(t,z,b)}\,,
\end{align}
where, as before, $k(t,z,b)$ denotes the normal $(z,(T-t)\Sigma)$ density function evaluated at~$b$.
If $\phi$ and $\psi$ are both sufficiently smooth to apply It\^o's formula, then we obtain
\begin{equation}\label{informed2}
\theta_t = \frac{1}{T-t}(\tzeta-Y_t) +\Sigma \nabla_y  \psi(t,Y_t,\tilde \zeta)\,.
\end{equation}
Here are some specific cases in which we can verify the regularity conditions assumed in Theorem~\ref{thm:riskaverse} and also verify the smoothness needed to obtain the formulas \eqref{eq:gp} and \eqref{informed2}.  We remark that case (ii) in the following is obtained from Caffarelli's contraction theorem  \citep{Caffarelli:2000:MonotonicityPropertiesOptimal}, which implies that $\nabla^2 \Gamma(y)$ and $\nabla \Gamma^2 (t,y)$ are bounded. The matrix $I$ in the following theorem is the $n \times n$ identity matrix.

%%%%%%%%%%%%%%%%%%%%%%%%%%%%%%%%%%%%%%%%%%%%%%%%%
%
% Theorem 4.2 (Risk-Averse Equilibrium with Assumptions on Fundamentals)
%
%%%%%%%%%%%%%%%%%%%%%%%%%%%%%%%%%%%%%%%%%%%%%%%%%

\begin{theorem}\label{thm:riskaverse2}
Assume that either (i) $F$ has bounded support or (ii) $F$ is absolutely continuous with respect to Lebesgue measure and its density is $\E^{-U(v)}$ for some $U$ satisfying $\nabla^2 U\geq \kappa I$ for some $\kappa>0$.  Then, Assumptions~\ref{integrability} and~\ref{assumeh0} hold, and there is a unique strong solution of \eqref{eqy}.  Therefore, the assumptions of Theorem~\ref{thm:riskaverse} hold.  Furthermore, $ \gamma(t,Y_t)=\nabla\phi(t,Y_t)$ is the vector of prices of risk, \eqref{informed2} is the equilibrium informed trading strategy, and
\begin{equation}\label{newdY2}
    \D Y_t = \Sigma \nabla \phi(t,Y_t)\,\D t + \D \hat Y_t
\end{equation}
where $\hat Y$ is a $(0,\Sigma)$--Brownian motion on $(\P,\cF^Y)$.
\end{theorem}

%A particular instance of case (ii) in Theorem~\ref{thm:riskaverse2} is a normal distribution.  We discuss normal distributions in the next section.  We end this section with an example of case (i) in Theorem~\ref{thm:riskaverse2} in which we find the physical distribution of $\tv$ numerically.

We provide an example of condition (i) in Theorem~\ref{thm:riskaverse2} in the online appendix.  It illustrates the role of the market makers' initial endowment $\beta$. In that example, we calculate the physical distribution $\hat F$ numerically.  A particular instance of case (ii) in Theorem~\ref{thm:riskaverse2} is a normal distribution.  We provide analytic results for normal distributions in the next section.

%%%%%%%%%%%%%%%%%%%%%%%%%%%%%%%%%%%%%%%%%%%%%%%%%
%
\section{Risk Aversion and Normal Distributions}\label{s:normal}
%
%%%%%%%%%%%%%%%%%%%%%%%%%%%%%%%%%%%%%%%%%%%%%%%%%

When the risk-neutral distribution of $\tv$ is normal, we can explicitly compute the equilibrium in Theorem~\ref{thm:riskaverse}, including the physical distribution, which is also normal.  The map from normal risk-neutral distributions $F$ to normal physical distributions $\hat F$ is bijective, so the equilibrium exists whenever the physical distribution is normal.  Moreover, it can be explicitly described in terms of the physical distribution.  We give formulas for the matrices $S$, $\Lambda$, and $A_t$ that appear in Theorem~\ref{thm:normal} at the end of this section.

%%%%%%%%%%%%%%%%%%%%%%%%%%%%%%%%%%%%%%%%%%%%%%%%%
%
% Theorem 5.1 (Normal distribution)
%
%%%%%%%%%%%%%%%%%%%%%%%%%%%%%%%%%%%%%%%%%%%%%%%%%

\begin{theorem}\label{thm:normal}
Consider any $\hat m \in \R^n$ and any symmetric positive-definite $n \times n$ matrix $\hat S$.  Set $m = \hat m - \alpha \hat S \beta$.  There exists a symmetric positive-definite matrix $S$ such that, if $F$ is the normal $(m,S)$ distribution function, then the assumptions of Theorem~\ref{thm:riskaverse} hold and the physical distribution $\hat F$ defined in Theorem~\ref{thm:riskaverse} is normal with mean $\hat m$ and covariance matrix $\hat S$.  There is a random vector $\tzeta$ and a symmetric positive-definite matrix $\Lambda$ such that the equilibrium informed trading strategy is
\begin{equation}\label{informedgaussian}
 \theta_t = \frac{1}{T-t}(\tzeta - Y_t)\,.
\end{equation}  
and the equilibrium pricing rule is 
\begin{equation}
    P_t = m + \Lambda Y_t\,.
\end{equation}
There is a nonrandom symmetric positive-definite matrix $A_{t}$ such that
\begin{align}
    \D Y_t & = - \Sigma^{1/2} A_{t} \Sigma^{-1/2}(Y_t-\beta)\,\D t + \D \hat Y_t\,,\label{corY}\\
    \D P_t &= - \Sigma^{-1/2}A_{t} \Sigma^{1/2}(P_t - m - \Lambda \beta)\,\D t + \Lambda\,\D \hat Y_t\,,\label{corP}
\end{align}
where $\hat Y$ is a $(0,\Sigma)$--Brownian motion on $(\P,\cF^Y)$.
Holding the physical distribution of $\tv$ fixed, denote the dependence of the matrices $S$, $\Lambda$, and $A_t$ on $\alpha$ by writing them as $S_\a$, $\Lambda_\a$, and $A_{\alpha t}$.  The matrices are all increasing in $\alpha$ in the sense that, if $\alpha'>\alpha$, then $S_{\alpha'}-S_\alpha$, $\Lambda_{\alpha'}-\Lambda_\a$ and $A_{\alpha't}-A_{\alpha t}$ for all $t$ are positive definite matrices. 
The equilibrium expected profit of the informed trader, conditional on $\tv$, is
\begin{equation}\label{cor:expprofit}
   \frac{1}{2} (\tv - \hat m - \a\hat S \beta)'\Lambda^{-1}(\tv-\hat m- \a\hat S \beta)  + \frac{1}{2}\tr(T\Sigma \Lambda)\,.
\end{equation}
The unconditional expected profit of the informed trader is
\begin{equation}\label{cor:expprofit2}
   \tr(T\Sigma\Lambda)  - \frac{\alpha}{2}\tr(T\Sigma \hat S) + \frac{\a^2}{2}\beta'\hat S \Lambda^{-1} \hat S \beta \,,
\end{equation}
and it is an increasing function of $\alpha$.  The expected loss of the noise traders is $\tr\left(T\Sigma \Lambda\right)$, and  the expected terminal wealth of market makers is 
$\beta'\hat m +
   \a\tr(T\Sigma \hat S)/2 -\a^2\beta'\hat S \Lambda^{-1} \hat S \beta/2$.
\end{theorem}

Theorem~\ref{thm:normal} provides explicit descriptions of phenomena that we expect to hold qualitatively for non-normal distributions as well.  For example, it illustrates the effect of market makers' inventories on equilibrium prices, the mean reversion in their inventories, and the consequent mean reversion in prices that result from risk aversion.  Market makers' inventories equal $\beta-Y_t$.  Equation~\eqref{corY} shows that the vector $\Sigma^{-1/2}(\beta-Y_t)$ is mean reverting to zero with the symmetric mean-reversion matrix $A_{t}$.  Likewise, equation~\eqref{corP} shows that the vector $\Sigma^{1/2}(P_t-m-\Lambda \beta)$ is mean reverting to zero also with mean-reversion matrix $A_{t}$.  %The sign of the expected change in $P$ is always the sign of $\beta-Y_t$; thus, risk premia are positive when market makers have positive inventories (modulo the off-diagonal elements of the matrix in \eqref{corP}).  The risk premium at date 0 is
%$\mye[\tv] - P = \alpha \hat S \beta$.
We have $P_t-m-\Lambda \beta = \Lambda (Y_t-\beta)$, so \eqref{corP} shows that risk premia depend on market maker inventories and generally (modulo the off-diagonal elements of the matrix in \eqref{corP}) have the same signs as inventories.  At date 0, inventories equal endowments $\beta$, and expected price changes from 0 to $T$ equal $\mye[\tv]-P_0 = \alpha \hat S \beta$.
This is the same formula that arises in a competitive CARA/normal model, except that, in the competitive model, the total number of shares outstanding appears in place of the market makers' inventory~$\beta$.
The risk premia and mean reversion go hand in hand with `excess volatility.' 
The excess volatility is manifested in quadratic variations being larger than variances.  The quadratic variation of $P$ over $[0,T]$ equals the 
risk-neutral covariance matrix $S$, because $P$ is a risk-neutral martingale.  Quadratic variations are therefore increasing in risk aversion.

Risk aversion has two effects on the informed trader's expected profit.  First, it makes the market less liquid in the sense that $\Lambda$ is larger, so $\Lambda^{-1}$ in \eqref{cor:expprofit} is smaller.  The reduction in liquidity reduces the informed trader's profit.  For example, the term 
$ (\tv - \hat m)'\Lambda^{-1}(\tv-\hat m)/2$
in \eqref{cor:expprofit} is decreasing in risk aversion.  However, there is another important effect: risk-averse market makers are willing to trade at prices different from expected values in order to reduce inventory risk, which is advantageous to the informed trader.  This can be seen, for example, in the term
$\alpha (\tv - \hat m)' \Lambda^{-1}\hat S \beta$
in \eqref{cor:expprofit}, which shows that the informed trader makes more money if her information is in the same direction as the market makers' inventory $\beta$---for example, if she wants to buy when market makers are already long.  The second effect is the most important, so, as the corollary states, the unconditional expected profit of the informed trader is increasing in $\alpha$, even if market makers start with zero inventory.

As in the risk-neutral model, we can say that the expected losses 
 $\tr\left(T\Sigma \Lambda\right)$ of noise traders, which are increasing in $\alpha$, are transfers to the informed trader.  However, when~$\beta$ is small, the informed trader makes less than  $\tr\left(T\Sigma \Lambda\right)$, because market makers must be compensated for the inventory risk that is created by making the market.  On the other hand, when the initial inventory $\beta$ is large, informed traders make more than  $\tr\left(T\Sigma \Lambda\right)$, because market makers are willing to incur expected losses to shed the inventory risk.
 
 The matrices $S$, $\Lambda$, and $A_t$ that appear in Theorem~\ref{thm:normal} are defined as follows.  
Diagonalize the symmetric positive-definite matrix $T\Sigma^{1/2}\hat S\Sigma^{1/2}$ as $ V' \hat D V$ where $\hat D$ is the diagonal matrix of eigenvalues and $V$ is the orthogonal matrix of eigenvectors.  
 Let $\hat d_1 ,\ldots,\hat d_n$ denote the eigenvalues.  For each $i$, define $d_i$ by 
\begin{align}\label{DDhat}d_i = \hat d_i + \frac{1}{2}\alpha \hat d_i\left(\alpha \hat d_i + \sqrt{\alpha^2\hat d_i^2 + 4\hat d_i}\right)  \quad \Rightarrow \quad \sqrt{d_i} = \frac{\alpha }{2}\hat d_i + \sqrt{\frac{\alpha^2 \hat d_i^2}{4} +  \hat d_i}
\,.
\end{align}
For each $t \in [0,T]$ and each $i=1,\ldots,n$, define
$$d_{it} = \frac{\alpha T \sqrt{d_i}}{T + (T-t)\alpha \sqrt{d_i}}\,.$$
Set $D = \text{diag}(d_1,\ldots,d_n)$ and $D_t = \text{diag}(d_{1t},\ldots,d_{nt})$.  The matrices in Theorem~\ref{thm:normal} are:
\begin{align}
S &=T^{-1}\Sigma^{-1/2} V'DV\Sigma^{-1/2}\,,\\
\Lambda &= T^{-1} \Sigma^{-1/2}V'D^{1/2}V\Sigma^{-1/2}\,,\\
A_t &= T^{-1}V'D_tV\,.
\end{align}
The matrices are increasing in $\alpha$ because each $d_i$ and each $d_{it}$ is increasing in $\alpha$.  

It is instructive to consider the univariate case.  Suppose $\tv$ is a scalar and, under the physical distribution, is normal with mean $\hat{m}$ and variance $\hat \sigma_v^2$.  Denote the cumulative variance $T \Sigma$ of noise trades by $\sigma_z^2$.  The risk-neutral distribution is normal with mean $m = \hat m - \alpha \beta \hat \sigma_v^2$ and variance $\sigma_v^2$.  After some simplification, we see from the above formulas that Kyle's lambda is $\lambda = \sigma_v/\sigma_z$, which is the same as Kyle's formula but using the risk-neutral standard deviation $\sigma_v$ instead of the physical standard deviation of $\tv$.  In terms of the physical standard deviation, we have 
$$ \lambda = \frac{\alpha}{2}\hat\sigma_v^2 + \sqrt{\frac{\alpha^2 \hat\sigma_v^4}{4} + \frac{\hat\sigma_v^2}{\sigma_z^2}}\,.$$
Keeping the physical distribution fixed, $\lambda$ is increasing in $\alpha$, so dealer aversion to inventory risk reduces market liquidity.  Because $\lambda$ is increasing in $\alpha$, the risk-neutral standard deviation $\sigma_v = \lambda \sigma_z $ is also increasing in $\alpha$.   The mean-reversion coefficient $A_t$ increases over time, starting at
$$\frac{1}{T} \cdot \frac{\alpha\lambda\sigma_z^2}{1+\alpha\lambda\sigma_z^2}$$
at date $t=0$ and rising to $\alpha\lambda\sigma_z^2/T$ at date $t=T$.
It is also increasing in $\alpha$ for each $t$.

%%%%%%%%%%%%%%%%%%%%%%%%%%%%%%%%%%%%%%%%%%%%%%%%%
%
\section{Risk Aversion and Options}\label{s:options}
%
%%%%%%%%%%%%%%%%%%%%%%%%%%%%%%%%%%%%%%%%%%%%%%%%%

We now apply the model with risk-averse market makers to study informed trading in options.  In particular, we investigate the extent to which information from the options market can be used to predict the return of the underlying asset.  We consider an asset and a European call option on the asset that matures at the announcement date~$T$.  This is the model studied by \citet{Back:1993:AsymmetricInformationOptions}, except that we allow market makers to be risk averse---and hence for the return of the underlying asset to be predictable---and we do not need to make special parametric assumptions to solve the model.  By put-call parity, the model is equivalent to one in which an underlying asset and a put option are traded or to one in which an underlying asset and a straddle are traded.  For the sake of brevity, we call the underlying asset a stock.  We use subscripts $s$ and~$o$ to denote `stock' and `option.'

We take the initial endowment of the market makers to be $\beta=0$, and we set absolute risk aversion to be $\alpha=0.2$.\footnote{The qualitative results do not depend on the magnitude of risk aversion.  Our choice of $\alpha=0.2$ is prompted by the following reasoning.  We have not specified the units of wealth, but since we will take the standard deviation of stock noise trading to be 2 in all examples in this section, a reasonable unit would be \$100 million.  An investor with an absolute risk aversion of 0.2 when wealth is measured in units of \$100 million would pay \$1,000 to avoid a coin toss for \$1,000,000.  This seems like a not unreasonable level of risk aversion.}   We assume the stock is lognormally distributed: $\log \tv_s = \mu + \sigma \tilde x$, where $\tilde x$ is a standard normal under the risk-neutral distribution.  We set $\sigma = 0.2$ and $\mu = \log (100) - \sigma^2/2$. In our numerical solution, which is described in the online appendix, we work with $\log \tv_s$ on a grid of $\mu \pm 3.3 \sigma$, so we can also view this as an example of a truncated lognormal distribution and hence as satisfying condition (i) of Theorem~\ref{thm:riskaverse2}.  The risk-neutral mean of $\tv_s$ is 100, and we set the strike of the option to be 100.  We take $T=1$, and we consider four different $\Sigma$ matrices---varying the sign of the correlation between stock and option orders and varying the relative standard deviations of stock and option orders---and solve the model for each.  We then simulate 10,000 sample paths for each choice of $\Sigma$. 
Figures~\ref{fig:1} and~\ref{fig:sdf_density} present aspects of the solution for one choice of $\Sigma$; Figures~\ref{fig:bidensity} and~\ref{fig:mmprofits} present the solution for two values of $\Sigma$; and Table~\ref{tab:regr} presents results for all four values of $\Sigma$.    The qualitative features shown in the figures are very similar for the other values of $\Sigma$.  

Figure~\ref{fig:1} illustrates the solution of the model at $t=0.5$.  The state vector is $Y_t$, which is the vector of stock and option order imbalances.  The option is in the money (the stock price is above 100) in the top-right part of each panel in Figure~\ref{fig:1}.  Panels (a) and (b) show that, when the option is in the money, the stock and option prices are approximately functions of the sum of stock and option imbalances.  This reflects the fact that the assets are informationally equivalent when the option is certain to finish in the money.  On the other hand, in the bottom left part of the panels, the option is out of the money.  When the option is deeply out of the money, the stock price is virtually independent of the option imbalance, reflecting the fact that option orders are almost certainly pure noise in that circumstance.  These features become more prominent as $t \rightarrow T$, and at $t=T$, the stock price depends solely on $y_s+y_o$ when the option is in the money, depends solely on $y_s$ when the option is out of the money, and is discontinuous in $(y_s,y_o)$ at the locus of points such that the option is exactly at the money.  

\begin{figure}
    \centering
    \includegraphics[scale=0.6]{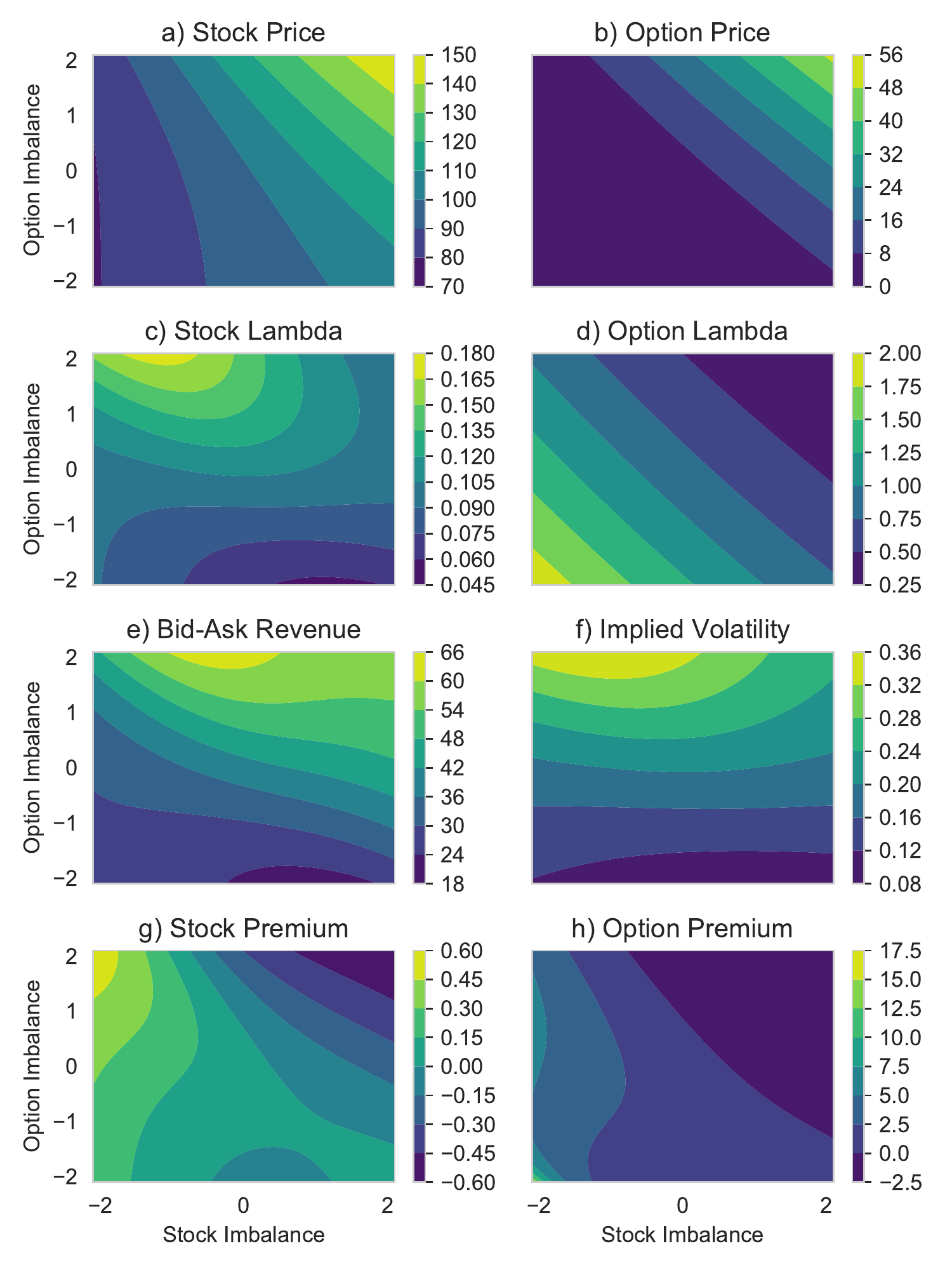}
    \caption{\textbf{Solution of the Model at $t=0.5$} \newline The figure presents properties of the model at the midpoint $t=0.5$ as a function of the state variables $Y_s=\,$ Stock Imbalance and $Y_o = \,$ Option Imbalance.  The plots are for $\Sigma = [[4,-2],[-2,4]]$.}
    \label{fig:1}
\end{figure}

Information flow in the market is affected in a qualitative way by the presence of option trading, which introduces a strong nonlinearity.  For example,  conditional variances can rise over time in nonlinear models, unlike linear Gaussian models.   In this model, the conditional variance of $\tv_s$  increases at many state/dates.  Uncertainty can also shrink very fast, depending on stock and option orders.  This phenomenon influences many of the variables plotted in Figure~\ref{fig:1}.  When the call has been purchased and the stock sold, market makers become very uncertain about the value of the stock.  It is quite likely that the trades in one of the assets are noise trades, but market makers do not know which is which.  On the other hand, when the call has been sold and the stock purchased, it is quite likely that the call will finish out of the money, but the stock value is not too low.  This narrows the range of possibilities considerably.  These two alternatives correspond to the top left and bottom right parts of the panels. This phenomenon also becomes more prominent as $t\rightarrow T$ and mirrors the discontinuity of the stock price in $(y_s,y_o)$ at $t=T$, which occurs roughly along the off-diagonal and is larger in the top left than in the bottom right in the $(y_s,y_o)$ plane.

The stock and option lambdas in Panels (c) and (d) of Figure~\ref{fig:1} are the diagonal elements of the $\Lambda$ matrix, each divided by its respective price, so they are in relative terms.  The option lambda is considerably larger than the stock lambda, consistent with the empirical fact that bid-ask spreads are higher and liquidity generally lower in options markets than in markets for the underlying assets.  The option lambda is decreasing in moneyness.  The stock lambda is influenced by the phenomenon discussed in the previous paragraph.  It is high in the top left where market makers are very uncertain about the stock value, and it is low in the bottom right where uncertainty is much lower.  Moreover, it is generally increasing in the option order imbalance.

Panel (e) in Figure~\ref{fig:1} plots  the bid-ask spread revenue received by market makers from noise traders per unit time.  This is $\tr(\Sigma \nabla^2 \Gamma(t,y))$.  Panel (f) plots the Black-Scholes implied volatility of the option.  Both plots have a similar pattern to the stock lambda: they are generally increasing in the option imbalance and in particular are high in the upper left and low in the bottom right.  

Panels (g) and (h) in Figure~\ref{fig:1}  plot the risk premia of the stock and option.  These are the elements of $\Lambda(t,y) \Sigma \nabla \phi(t,y)$, each divided by its respective price, so they are the expected rates of return per unit time.  Both are generally inversely related to the stock price, though the patterns are complex.  An interesting feature of the stock premium is that it is especially high in the upper left, where the implied volatility is also high.  This induces a positive correlation between the implied volatility and subsequent stock returns, as we will see.

Figure~\ref{fig:sdf_density} presents some data from a simulation of the model, using the same $\Sigma$ as in Figure~\ref{fig:1}.  We simulate the path of $Y$ and $\D P_t = \Lambda (t,Y_t)\,\D Y_t$, using the physical dynamics of $Y$ given in Theorem~\ref{thm:riskaverse2}.  The terminal value of $P$ is $P_T = (\tv_s, \tv_o)$, where $\tv_o = (\tv_s-K)^+$.  Panel (a) presents the realized values of the SDF $\E^{-\alpha \tw}/\mye[\E^{-\alpha \tw}]$, replacing the expectation with the mean across simulations, plotted against the realized values of $\tv_s$.  We expect the SDF to be high for extreme values of $\tv_s$, because the informed trader makes more money and market makers typically lose money when the informed trader has unexpected information.  This same phenomenon appears in Theorem~\ref{thm:normal}.  It is somewhat surprising that the SDF is also often high for values of $\tv_s$ near the option strike.  This occurs when there are large buy orders and option sell orders early during the trading period.  These push the market to the bottom right in the panels in Figure~\ref{fig:1}, where uncertainty is low.  Competition between dealers is then much like risk-neutral competition, pushing down the lambdas and the bid-ask revenue from noise traders.  Consequently, the bid-ask revenue is insufficient to cover the losses to the informed trader from the initially overvalued option.  

Panel (b) of Figure~\ref{fig:sdf_density} shows the assumed lognormal risk-neutral density and a kernel estimate of the physical density of $\tv_s$, based on the simulated values of $P_T$.  The risk-neutral  density is the physical density multiplied by the SDF, so high values of the SDF produce high risk-neutral probabilities.  Because the SDF is high for extreme values of $\tv_s$, the risk-neutral distribution has extra weight in the tails and a larger variance than the physical distribution, as in  Theorem~\ref{thm:normal}.  The high SDF near the option strike  also adds weight to that part of the distribution.  If the physical distribution were unimodal, this would produce a spike in the risk-neutral distribution.  However, we start with a lognormal risk-neutral distribution, so the added weight in the risk-neutral distribution is possible only if the physical density is small near the option strike.  Consequently, the physical density that is consistent with the given risk-neutral density is bimodal.  This particular feature of the model seems to be an artifact of assuming there is only a single option traded.  Trading multiple strikes should smooth out the spikes in the risk-neutral density or equivalently the troughs in the physical density.

\begin{figure}
    \centering
    \includegraphics[scale=0.6]{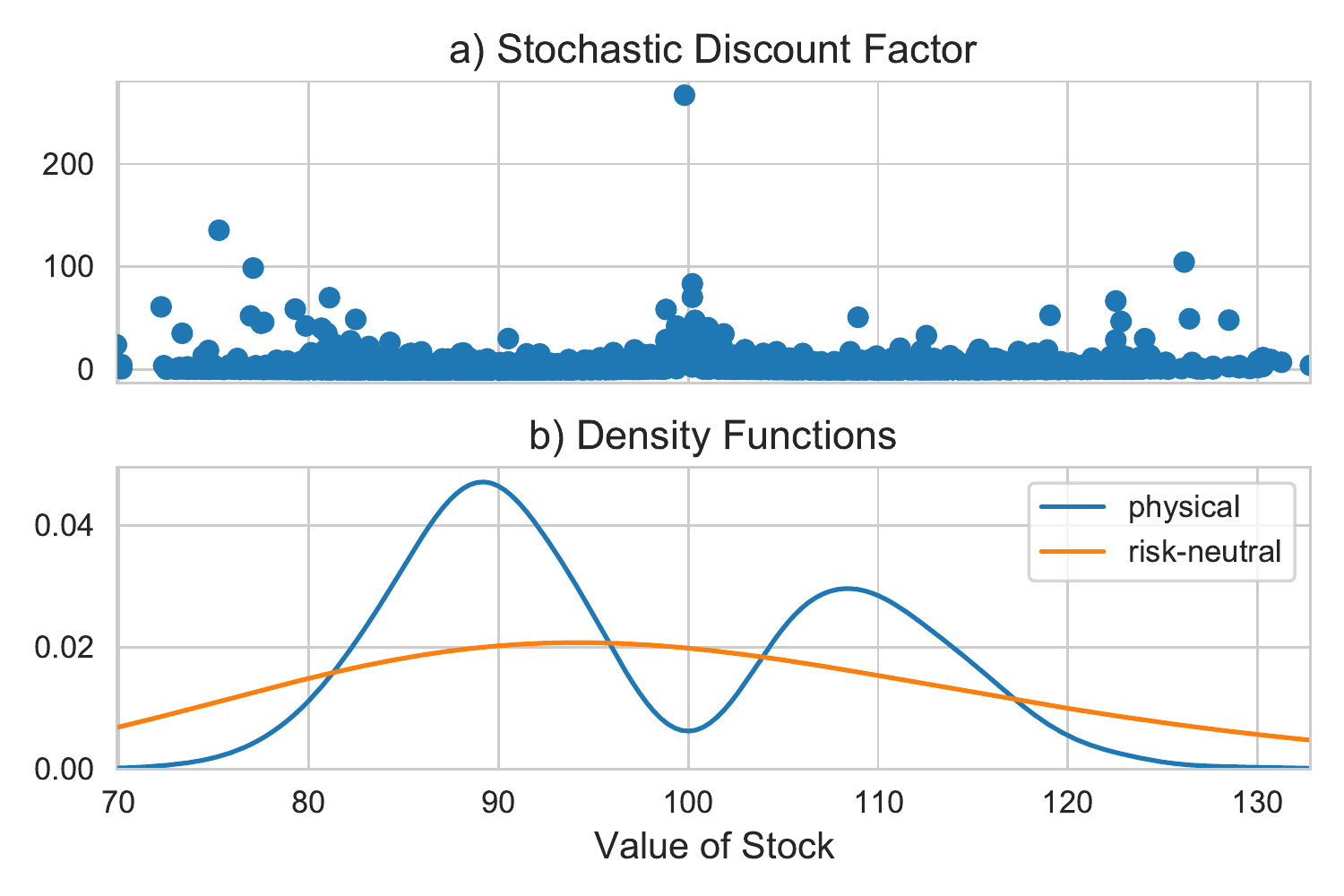}
    \caption{\textbf{Risk Neutral and Physical Densities} \newline
    Panel (a) is a scatter plot of the SDF versus realized values of $\tv_s$.  Panel (b) presents a kernel estimate of the physical density based on realized values of $\tv_s$ and also presents the assumed lognormal risk-neutral density.  The scatter plot and kernel estimate are based on 10,000 simulations with $\Sigma = [[4,-2],[-2,4]]$.}
    \label{fig:sdf_density}
\end{figure}

Figure~\ref{fig:bidensity} shows the physical distribution of $Y_T$ for two choices of $\Sigma$.  Panel (a) uses the same value of $\Sigma$ as in Figures~\ref{fig:1} and~\ref{fig:sdf_density}.  Panel (b) uses the same $\Sigma$ except that the correlation of $Z_s$ and $Z_o$ is switched from negative to positive.  The figure shows that the correlation of net imbalances $Y_s$ and $Y_o$ from date 0 to $T$ is negative, regardless of the correlation of noise trades.  This means that market makers usually end up with hedged positions---long the stock if they are short calls and long calls if they are short the stock.  This is an intuitive consequence of risk aversion.  One can see from the figure that market makers are short the call and long the stock at the modes of the distributions.  The same is true for the means of the distributions, so on average market makers in this model sell options and use the stock to hedge, consistent with what actual option market makers usually do (for example, \citet{ni_etal} show that option market makers  on average hold portfolios with negative gammas, meaning that they are net short options).

\begin{figure}
    \centering
    \includegraphics[scale=0.6]{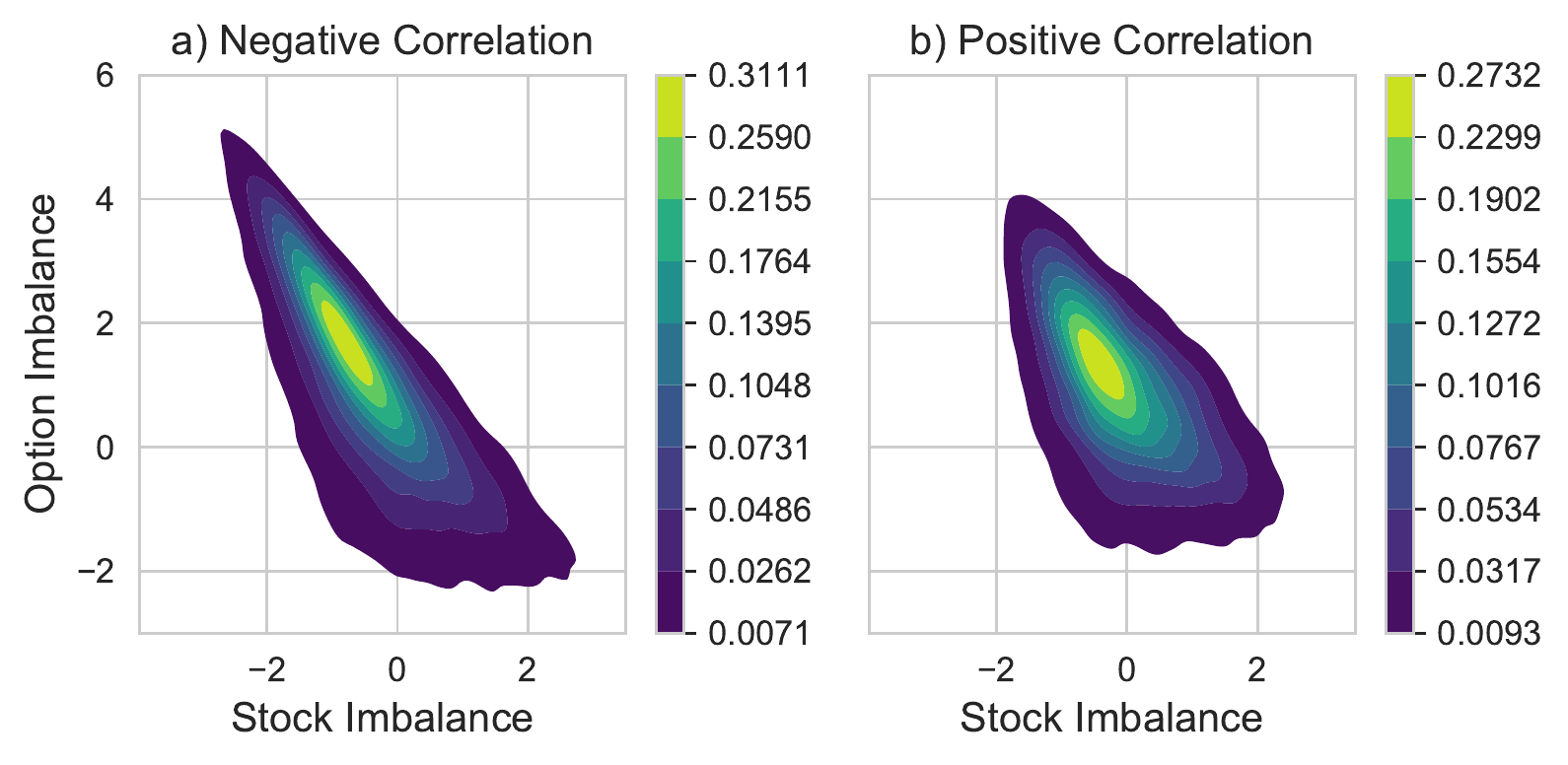}
    \caption{\textbf{Physical Density of Order Imbalances at date $t=1$}\newline
    The plots are kernel estimates of the bivariate density of $Y_s = \,$ Stock Imbalance and $Y_o=\,$ Option Imbalance at the terminal date $t=1$.  Panel (a) is based on 10,000 simulations with $\Sigma = [[4,-2],[-2,4]]$, and Panel (b) is based on 10,000 simulations with $\Sigma = [[4,2],[2,4]]$.}
    \label{fig:bidensity}
\end{figure}

Figure~\ref{fig:mmprofits} plots market maker profits as a function of $Y_T$ for the same two values of $\Sigma$ as in Figure~\ref{fig:bidensity}.  Market maker profits depend on the path of $Y$, not just on $Y_T$, so the figure shows an estimate of expected market maker profits conditional on $Y_T$.  The region in which profits are highest is the upper left, which, as discussed in connection with Figure~\ref{fig:1}, is the area in which market makers are most uncertain about the value of the stock.  The combination of stock sells and option buys amplifies uncertainty for market makers.  To compensate for this extra risk, price impacts and bid-ask revenue are higher, as shown in Figure~\ref{fig:1}, and market maker profits are higher on average.

\begin{figure}
    \centering
    \includegraphics[scale=0.6]{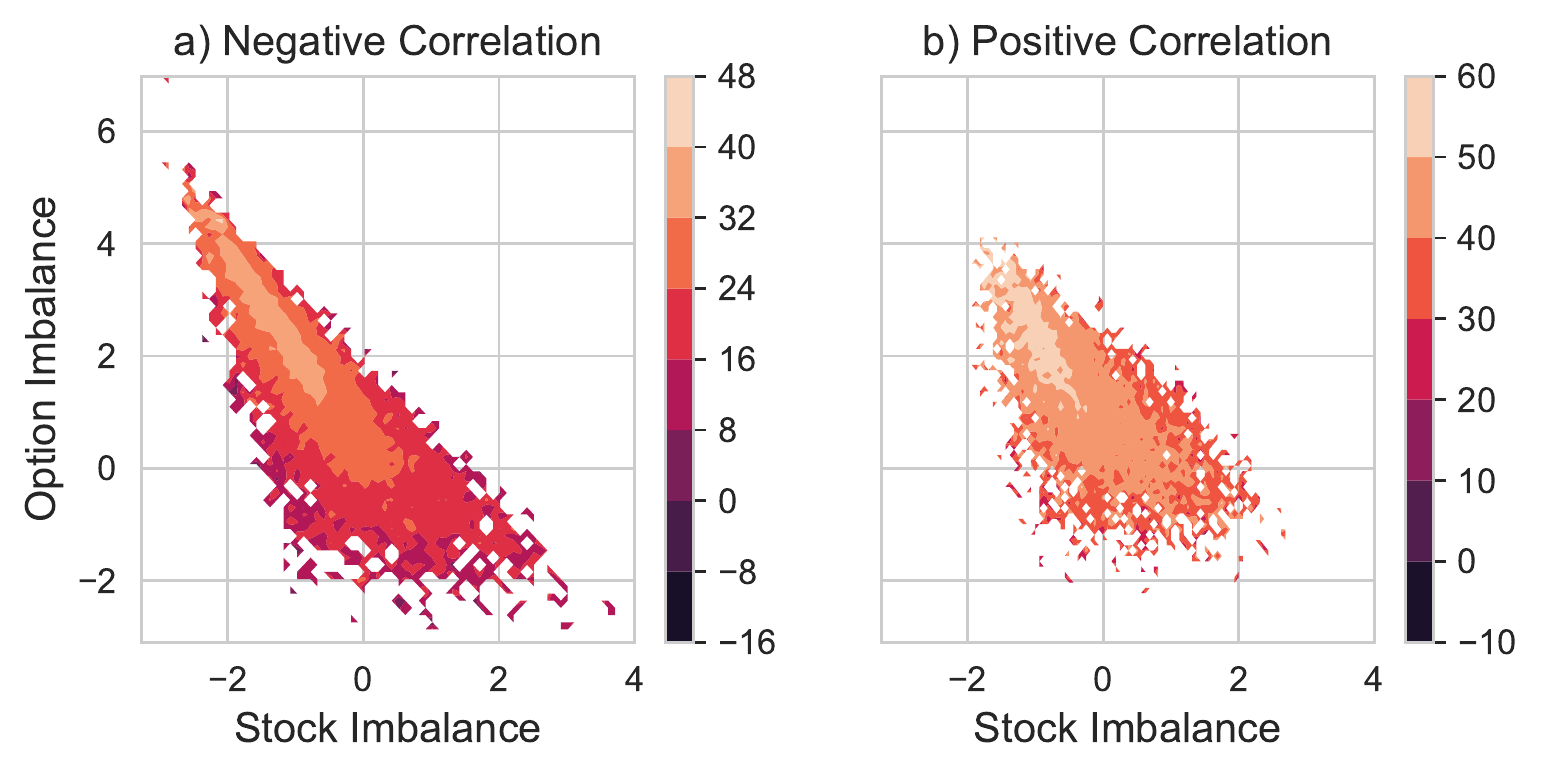}
    \caption{\textbf{Market Maker Profits}\newline
    Realized market maker profits are averaged within bins in a $100 \times 100$ grid.
     Panel (a) is based on 10,000 simulations with $\Sigma = [[4,-2],[-2,4]]$, and Panel (b) is based on 10,000 simulations with $\Sigma = [[4,2],[2,4]]$.}
    \label{fig:mmprofits}
\end{figure}

We can use the simulated model to estimate moments and conditional moments for which we do not have analytic expressions.  As an illustration, we look at the correlation between the implied volatility and future stock returns.  \citet{anetal} document empirically that higher implied volatilities predict higher future stock returns.  We find the same correlation in our model: the conditional risk premium is higher when the implied volatility is higher.  This is a robust finding across all values of $\Sigma$ that we examine.  It is driven by the fact that implied volatilities are highest in the upper left corners of the panels in Figure~\ref{fig:1}, where the stock has been sold and the option purchased.  The stock risk premium is a decreasing function of stock orders, and is also especially high in the upper left corner in Figure~\ref{fig:1}; hence, the implied volatility and risk premium are positively correlated.   Table~\ref{tab:regr} reports regressions of the risk premium on the implied volatility at $t=0.5$ for all four values of $\Sigma$.\footnote{We look only at options with time remaining to maturity equal to 0.5 to mirror An et al.'s study of options with a fixed time to maturity (30 days).  The implied volatility is initially 20\% for all observations, so regressing on the change is equivalent to regressing on the level.}  The coefficient is positive and significant in all of the univariate regressions.

\begin{table}
    \caption{Regression of Stock Risk Premium on Implied Volatility at $t=0.5$}
    \label{tab:regr}
    
    \footnotesize\vspace{1em}
    The stock risk premium (defined as the stock-element of the vector $\Lambda(t,y) \Sigma \nabla \phi(t,y)$  divided by the stock price) is regressed at the midpoint $t=0.5$ on the change in the implied volatility from $t=0$ to $t=0.5$. The data points are independent simulations of the model.  Regressions (2), (4), (6), and (8) include a control for the continuously compounded stock return from $t=0$ to $t=0.5$. Under each coefficient is the 95\% confidence interval (between square brackets) and standard error.

    \begin{tabularx}{\textwidth}{lRRRR}\toprule
    & \multicolumn{2}{c}{$\Sigma=[4,-2 ; -2, 4]$} & \multicolumn{2}{c}{$\Sigma=[4, -1; -1, 1]$} \\
        \cmidrule(r){2-3}\cmidrule(lr){4-5}
& \multicolumn{1}{c}{(1)} & \multicolumn{1}{c}{(2)} & \multicolumn{1}{c}{(3)} & \multicolumn{1}{c}{(4)} \\
\midrule
Intercept & -0.003 & 0.075 & 0.007 & 0.037\\
 & {}[-0.011, 0.005] & {}[0.072, 0.078] & {}[0.002, 0.012] & {}[0.036, 0.038]\\
 & s.e. = 0.004 & s.e. = 0.001 & s.e. = 0.003 & s.e. = 0.001\\[1.1ex]
$\Delta$ Implied vol & 0.444 & 0.491 & 0.354 & 0.547\\
 & {}[0.384, 0.503] & {}[0.472, 0.511] & {}[0.277, 0.430] & {}[0.527, 0.568]\\
 & s.e. = 0.030 & s.e. = 0.010 & s.e. = 0.039 & s.e. = 0.010\\[1.1ex]
Stock return &  & -2.187 &  & -1.774\\
 &  & {}[-2.202, -2.172] &  & {}[-1.783, -1.764]\\
 &  & s.e. = 0.008 &  & s.e. = 0.005\\
\midrule
Obs. & 10000 & 10000 & 10000 & 10000\\
R$^2$ & 0.021 & 0.894 & 0.008 & 0.932\\
\midrule
 & \multicolumn{2}{c}{$\Sigma=[4, 2; 2, 4]$} & \multicolumn{2}{c}{$\Sigma=[4, 1; 1, 4]$} \\
        \cmidrule(r){2-3}\cmidrule(lr){4-5}
& \multicolumn{1}{c}{(5)} & \multicolumn{1}{c}{(6)} & \multicolumn{1}{c}{(7)} & \multicolumn{1}{c}{(8)}\\
\midrule
Intercept & 0.041 & 0.346 & 0.037 & 0.183\\
 & {}[0.035, 0.048] & {}[0.343, 0.348] & {}[0.032, 0.041] & {}[0.182, 0.184]\\
 & s.e. = 0.003 & s.e. = 0.001 & s.e. = 0.002 & s.e. = 0.001\\[1.1ex]
$\Delta$ Implied vol & 0.381 & -0.603 & 0.253 & -0.188\\
 & {}[0.280, 0.482] & {}[-0.630, -0.575] & {}[0.135, 0.372] & {}[-0.217, -0.160]\\
 & s.e. = 0.052 & s.e. = 0.014 & s.e. = 0.060 & s.e. = 0.014\\[1.1ex]
Stock return &  & -2.360 &  & -1.830\\
 &  & {}[-2.373, -2.347] &  & {}[-1.839, -1.821]\\
 &  & s.e. = 0.007 &  & s.e. = 0.004\\
\midrule
Obs. & 10000 & 10000 & 10000 & 10000\\
R$^2$ & 0.005 & 0.928 & 0.002 & 0.943\\
\bottomrule
    \end{tabularx}

\end{table}

Table~\ref{tab:regr} also reports regressions of the risk premium on the implied volatility and the prior stock return.  This is a partial replication of Table VI of \citet{anetal}, which controls for prior returns (`short term reversal') and other stock return predictors studied in the literature.  The coefficient on the prior return is consistently negative in Table~\ref{tab:regr}, as it is in An et al.'s Table VI.  This is a consequence of excess volatility and mean reversion as discussed previously.  The sign of the coefficient on the implied volatility, unlike the univariate coefficient, is sensitive to the model parameters.  It is positive (as in the data) when the correlation between stock and option noise trades is negative (Columns 2 and 4) but negative when the correlation between stock and option noise trades is positive (Columns 6 and 8).  The reason that the sign varies across these cases in the multivariate regressions is that the correlation between the regressors changes sign across the cases.  The prior return and implied volatility are positively correlated in Columns 2 and 4 (thus, the multivariate regression coefficient is larger than the univariate coefficient) but negatively correlated in Columns 6 and 8 (thus, the multivariate regression coefficient is smaller than the univariate coefficient and in fact is negative).  Additional empirical investigation of this phenomenon would be useful.

%%%%%%%%%%%%%%%%%%%%%%%%%%%%%%%%%%%%%%%%%%%%%%%%%
%
\section{Conclusion}
%
%%%%%%%%%%%%%%%%%%%%%%%%%%%%%%%%%%%%%%%%%%%%%%%%%

The solution of the Kyle model is much simpler when described in terms of optimal transport theory.  By recognizing this fact, we are able to extend the solution to a much broader class of models.  The extension to risk-averse market makers combines two important literatures that have heretofore developed primarily in parallel: the adverse selection and the inventory risk explanations of market illiquidity.  We are able to quantify the contribution each makes to illiquidity.  As one example of the applicability of this new approach, we derived a novel result regarding the predictive power of implied volatilities for stock returns.  We anticipate that many additional applications will be made in the future.   

%%%%%%%%%%%%%%%%%%%%%%%%%%%%%%%%%%%%%%%%%%%%%%%%%
%
\begin{appendix}
\section*{Proofs}
%%%%%%%%%%%%%%%%%%%%%%%%%%%%%%%%%%%%%%%%%%%%%%%%%

%%%%%%%%%%%%%%%%%%%%%%%%%%%%%%%%%%%%%%%%%%%%%%%%%
%
% Lemma A.1 (Integrability)
%
%%%%%%%%%%%%%%%%%%%%%%%%%%%%%%%%%%%%%%%%%%%%%%%%%
Some proofs rely on Lemmas~A.1 and~A.2.  Proofs of those are provided in the online appendix.

\begin{lemmaA1}\label{lem_A1}
Let $\Gamma$ be a Brenier potential as described in Theorem~\ref{thm_brenier}, where $\|\tv\|^2$ is integrable.
For each $y \in \R^n$, the random variable $\Gamma(Z_T)$ and random vector $\nabla \Gamma(Z_T)$ are integrable.  Furthermore, for each $y\in \R^n$ and each $t \in (0,T)$, the random variable $\Gamma(t,y+Z_T-Z_t)$ and random vector $\nabla \Gamma(t,y+Z_T-Z_t)$ are integrable.  
\end{lemmaA1} 

%%%%%%%%%%%%%%%%%%%%%%%%%%%%%%%%%%%%%%%%%%%%%%%%%
%
% Lemma A.2 (General Mixing)
%
%%%%%%%%%%%%%%%%%%%%%%%%%%%%%%%%%%%%%%%%%%%%%%%%%

\begin{lemmaA2}\label{lem_A2}
Let $L$ be an absolutely continuous distribution function on $\R^n$ and $F$ a second distribution function on $\R^n$.  Set $\mu = G \otimes F$ and assume that $L$ and $F$ have finite second moments.  Let $\Gamma$ be the Brenier potential such that $\nabla \Gamma$ transports $L$ to $F$.
Then, there exists a function $f:\R^{2n}\to \R^n$ such that
 $\mu(\{(u,v) \mid \nabla \Gamma(f(u,v))=v\}) = 1$, and
 $\mu(\{(u,v) \mid f(u,v) \le a\}) = L(a)$ for all $a \in \R^n$.
\end{lemmaA2} 

%%%%%%%%%%%%%%%%%%%%%%%%%%%%%%%%%%%%%%%%%%%%
%
% Proof of Theorem 3.1 (Brenier)
%
%%%%%%%%%%%%%%%%%%%%%%%%%%%%%%%%%%%%%%%%%%%%

\begin{proof}[Proof of Theorem~\ref{thm_brenier}]
The first statement can be found in \citet[Proposition 3.1]{Brenier:1991:PolarFactorizationMonotone}, with the exception that Brenier states that $\Gamma$ is unique up to an additive constant rather than specifying $\mye[\Gamma(Z_T)] = 0$.  Clearly, we can identify the constant uniquely with the requirement $\mye[\Gamma(Z_T)] = 0$ provided $\Gamma(Z_T)$ is integrable.  The second statement is also given in  
\citet[][Proposition 3.1]{Brenier:1991:PolarFactorizationMonotone}.  The integrability of $\Gamma$ and the last statement of the theorem are consequences of Lemmas~A.1 and~A.2. 
\end{proof}

%%%%%%%%%%%%%%%%%%%%%%%%%%%%%%%%%%%%%%%%%%%%%%%%%
%
% Proof of Lemma 3.1 (dGamma)
%
%%%%%%%%%%%%%%%%%%%%%%%%%%%%%%%%%%%%%%%%%%%%%%%%%

\begin{proof}[Proof of Lemma~\ref{lem_dGamma}]
See the online appendix.
\end{proof}

%%%%%%%%%%%%%%%%%%%%%%%%%%%%%%%%%%%%%%%%%%%%
%
% Proof of Theorem 3.2 (Risk-Neutral Equilibrium)
%
%%%%%%%%%%%%%%%%%%%%%%%%%%%%%%%%%%%%%%%%%%%%

\begin{proof}[Proof of Theorem~\ref{thm:riskneutral}]
Fix a trading strategy $X$ of the informed trader.  Because $X$ is a continuous semimartingale, we have $X=A+L$ where $A$ is a continuous finite-variation process, and $L$ is a continuous local martingale.  Moreover, $\D L_t = C_t \,\D Z_t$ for some matrix-valued process $C$.  Thus,
$\D Y_t = \D A_t + (C_t+I)\,\D Z_t$.
By It\^o's formula, substituting $\nabla \Gamma = H$, we have
$$\D \Gamma (t,Y_t) = \frac{\partial \Gamma (t,Y_t)}{\partial t}\,\D t + H(t,Y_t)'\,\D Y_t + \frac{1}{2}\tr \big( (C_t+I)\Sigma(C_t+I)'\nabla^2 \Gamma(t,Y_t)\big)\,\D t\,.
$$
Because $\Gamma$ is defined via the expectation \eqref{Gammaty}
and is smooth thanks to Lemma~\ref{lem_dGamma}, it solves the heat equation 
\begin{align}\label{pde:gamma}
\frac{\partial \Gamma}{\partial t}+\frac{1}{2}\tr\left(\Sigma \nabla \Gamma^2\right)=0\,.
\end{align}
Furthermore, a direct computation yields
$$\D \sum_{i=1}^n \langle  X^i, H^i(\cdot,Y_\cdot)\rangle_t={\tr }\big(C_t\Sigma(I+C_t)'{\nabla^2\Gamma(t,Y_t)}\big)\,\D t\,.$$
Hence,
\begin{multline*}
   \frac{1}{2}\tr \big( (C+I)\Sigma(C+I)'\nabla^2 \Gamma\big)\,\D t = \tr (C\Sigma\nabla^2\Gamma)\,\D t + \frac{1}{2} \tr (C\Sigma C' \nabla^2\Gamma)\,\D t + \frac{1}{2} \tr (\Sigma \nabla^2\Gamma) \, \D t\\
   = \D \sum_{i=1}^n \langle  X^i, H^i(\cdot,Y_\cdot)\rangle_t - \frac{1}{2} \tr (C\Sigma C' \nabla^2\Gamma)\,\D t - \frac{\partial \Gamma}{\partial t}\,\D t
\end{multline*}
Making this substitution, we obtain
$$\D \Gamma (t,Y_t) = H(t,Y_t)'\,\D Y_t + \D \sum_{i=1}^n \langle  X^i, H^i(\cdot,Y_\cdot)\rangle_t - \frac{1}{2} \tr \big(C_t\Sigma C_t' \nabla^2\Gamma(t,Y_t)\big)\,\D t\,.
$$
Integrating and rearranging and using the fact that $\Gamma(0,0)=0$ yields
\begin{multline}\label{app2:profit}
   \int_0^T (\tv-H(t,y_t))'\,\D Y_t - \sum_{i=1}^n \langle  X^i, H^i(\cdot,Y_\cdot)\rangle_T \\ = \tv'Y_T - \Gamma(Y_T) - \frac{1}{2}\tr\left(\int_0^T C_t\Sigma C_t' \nabla^2\Gamma(t,Y_t)\,\D t\right)\,.
\end{multline}
Under the restriction \eqref{restriction} for the strategy $X$, and using the independence and integrability of $\tv$ and $Z$, 
$$\mye\int_0^T (\tv - H(t,Y_t))'\,\D Z_t = 0\,.$$
Therefore, taking expectations in \eqref{app2:profit} yields
\begin{multline}\label{app2:expprofit}
   \mye\left[\int_0^T (\tv-H(t,Y_t))'\,\D X_t - \sum_{i=1}^n \langle  X^i, H^i(\cdot,Y_\cdot)\rangle_T \right]\\ = \mye\left[\tv'Y_T - \Gamma(Y_T) - \frac{1}{2}\tr\left(\int_0^T C_t\Sigma C_t' \nabla^2\Gamma(t,Y_t)\,\D t\right)\right]\,.
\end{multline}
The left-hand side is the definition of the informed trader's expected profit.  On the right-hand side, the trace is nonnegative due to the convexity of $\Gamma$.  Therefore, the expected profit of the informed trader is bounded above by $\mye[\tv'Y_T - \Gamma(Y_T)]$ and the bound is achieved by any finite-variation strategy such that, almost surely, 
$Y_T\in \argmax_{y\in \R^n}{\tv'y-\Gamma(y)}$, equivalently, $\nabla \Gamma (Y_T)=\tv$.  The strategy \eqref{brownbridge} implies
$$\D Y_t=\frac{\tzeta-Y_t}{T-t}\,\D t+\D Z_t\,,$$
which implies $Y_T=\tzeta$; hence $\nabla \Gamma(Y_T)=\tv$. The strategy achieves the maximum expected profit of
$$\mye[\sup \; \{ \tv'y - \Gamma(y) \mid y \in \R^n\}] = \mye[\Gamma^*(\tv)]\,.$$
To show that it is optimal, it remains only to show that it satisfies the restriction \eqref{restriction}.  First, we turn to the other equilibrium condition. 

Given the fact that the unconditional distribution of $\tilde \zeta$ is $G$, the strategy \eqref{brownbridge} implies that $Y$ is $(0,\Sigma)$-Brownian motion on its own filtration. Thus, conditionally on $\cF^Y_t$, $\tilde \zeta=Y_T$ has the normal $(Y_t,(T-t)\Sigma)$ distribution.  Consequently, 
$$H(t,Y_t)=\mye[\nabla \Gamma(Y_T)|\cF_t^Y]  = \mye[\tilde v|\cF_t^Y]\,.$$
This establishes the equilibrium condition \eqref{eq1}.

The quadratic variation of the process $\int_0^\cdot H(t,Y_t)'\,\D Z_t$ at $T$ is bounded above by a constant times $\int_0^T \| H(t,Y_T) \|^2\,\D t$.  Since the condition \eqref{eq1} holds and $\tv$ is square-integrable, $$\mye\int_0^T \| H(t,Y_T) \|^2\,\D t = \mye \int_0^T \parallel \mye[\tilde v|\cF_t^Y]\parallel^2 \,\D t < \infty\,.$$ 
%$ \int_0^\cdot H(t,X_t+Z_t)'\,\D Z_t$ satisfies
%\begin{align*}
%\mye\left[\left\langle \int_0^\cdot H(t,X_t+Z_t)'\,\D Z_t\right\rangle_T\right]&\leq C\mye \int_0^T \parallel H(t,Y_t)\parallel^2 \,\D t\\
%&= C \mye \int_0^T \parallel \mye[\tilde v|\cF_t^Y]\parallel^2 \,\D t < \infty
%\end{align*}
Thus, $\int_0^\cdot H(t,Y_t)'\,\D Z_t$ is a square integrable martingale and the restriction \eqref{restriction} holds. This completes the proof of optimality of the informed trading strategy, so both equilibrium conditions hold.

To verify the formula \eqref{valuefunction} for the value function, note that, by following the same reasoning as before but starting the integration of $\D \Gamma$ at $t$, we obtain 
\begin{multline*}
   \int_t^T (\tv-H(t,y_t))'\,\D Y_t - \int_t^T \D \sum_{i=1}^n \langle  X^i, H^i(\cdot,Y_\cdot)\rangle_t \\ = \Gamma(t,Y_t) + \tv'(Y_T-Y_t) - \Gamma(Y_T) - \frac{1}{2}\tr\left(\int_t^T C_t\Sigma C_t' \nabla^2\Gamma(t,Y_t)\,\D t\right)\,.
\end{multline*}
Taking expectations as before, we see that the conditional expected profit is bounded above by 
$$\Gamma(t,Y_t) - \tv'Y_t + \mye \left[\tv'Y_T - \Gamma(Y_T) \mid \cF^Y_t\right]\,.$$
Substituting 
$\sup\, \{ \tv'y - \Gamma(y) \mid y \in \R^n \} = \Gamma^*(\tv)$
yields the claim.

The definition \eqref{price1} of $H$ implies that the heat equation \eqref{heateqn} holds for each element $H_i$.   Combining this with It\^o's formula and the fact that $H(t,y) = \nabla \Gamma (t,y)$, we obtain
$\D P_t = \Lambda_t\,\D Y_t$ where $\Lambda_t = \nabla^2\Gamma(t,Y_t)$.  Thus, $\Lambda_t$ is symmetric, and the convexity of $\Gamma(t,y)$ implies that it is positive semidefinite. 
\end{proof}

%%%%%%%%%%%%%%%%%%%%%%%%%%%%%%%%%%%%%%%%%%%%
%
% Proof of Corollary 3.1
%
%%%%%%%%%%%%%%%%%%%%%%%%%%%%%%%%%%%%%%%%%%%%

\begin{proof}[Proof of Corollary~\ref{cor21}]
When $X$ has finite variation, \eqref{app2:profit} implies
$$\int_0^T (\tv-P_t)'\,\D Y_t = \tv'Y_T - \Gamma(Y_T)\,.$$
Substituting this into \eqref{dealerprofits} yields the first formula for dealer profits.  Optimization by the informed trader and the definition of the convex conjugate then yields the second formula.
\end{proof}

%%%%%%%%%%%%%%%%%%%%%%%%%%%%%%%%%%%%%%%%%%%%
%
% Proof of Lemma 4.1
%
%%%%%%%%%%%%%%%%%%%%%%%%%%%%%%%%%%%%%%%%%%%%

\begin{proof}[Proof of Lemma~\ref{gmarkov}]
See the online appendix.
\end{proof}

%%%%%%%%%%%%%%%%%%%%%%%%%%%%%%%%%%%%%%%%%%%%
%
% Proof of Lemma 4.2 (Mixing)
%
%%%%%%%%%%%%%%%%%%%%%%%%%%%%%%%%%%%%%%%%%%%%

\begin{proof}[Proof of Lemma~\ref{lem_mix}]
This follows directly from Lemma A.2.
\end{proof}

%%%%%%%%%%%%%%%%%%%%%%%%%%%%%%%%%%%%%%%%%%%%
%
% Proof of Theorem 4.1 (Risk Averse Equilibrium)
%
%%%%%%%%%%%%%%%%%%%%%%%%%%%%%%%%%%%%%%%%%%%%

\begin{proof}[Proof of Theorem~\ref{thm:riskaverse}]
{\it Step 1: The physical mean of $\tv$.}
From the definitions, we have
$$\mye^\P\big[|\tv|\big] = \mye^{\hat \nu}\big[|\nabla \Gamma (Z_T)|\big] = \mye^\nu \big[\xi_T |\nabla \Gamma (Z_T)|\big]\,,$$
which is finite because $\xi_T$ is $\nu$--square integrable by assumption, and because $\nabla \Gamma$ transports $G$ to $F$ which has finite second moments by assumption.

{\it Step 2: Filtering problem.}
If the informed trader uses the strategy \eqref{newtheta}, then $Y$ satisfies \eqref{eqy}. 
By assumption, this equation admits a unique strong solution on $[0,T]$. 
Thanks to \citet[Proposition 1]{yamada1971uniqueness}, this SDE also admits a unique martingale solution on $[0,T]$. Additionally, the domain of the infinitesimal generator of $(\tilde \zeta, Y)$ is the set of twice continuously differentiable functions which is dense in the set of continuous functions. Denoting $\hat p(t,a)$ the density of $\tilde \zeta$ given $\cF^Y_t$, \citet[Proposition 2.2, Theorem 3.3]{kurtz1988unique} leads to the uniqueness of martingale solutions for the filtering equation on $[0,T]$ which is
\begin{align*}
    \frac{\D \hat p(t,a)}{\hat p(t,a)}=&\left( \nabla_z  \log h(t,Y_t,a)-\int\nabla_z  \log h(t,Y_t,\tilde a)\hat p(t,\tilde a)\D \tilde a\right)'\\
    &\times \left(\D Y_t-\Sigma\gamma(t,Y_t)\D t-\Sigma \int\nabla_z  \log h(t,Y_t,\tilde a)\hat p(t,\tilde a)\D \tilde a\D t\right).
\end{align*}
Due to the choice of $\tilde \zeta$, the initial condition of this equation is $\hat p(0,a)=h(0,0,a)$.
Under $\hat \nu$, 
$\D Z_t-\Sigma\gamma(t,Z_t)\D t$
defines a martingale on $[0,T)$.
Given that $h$ is the $\hat \nu$ transition densities of $Z$, $h(t,Z_t,a)$
is a $\hat \nu$ martingale. 
Due to our assumption, the measures $h(t,Z_t,\cdot)$ satisfies
$$\D h(t,Z_t,a)=\nabla_z h(t,Z_t,a)'(\D Z_t-\Sigma\gamma(t,Z_t)\D t)$$
with initial condition $ h(0,0,a)=\hat p(0,a)$.
Thus, we have 
$$\D h(t,Y_t,a)=\nabla_z h(t,Y_t,a)'(\D Y_t-\Sigma \gamma(t,Y_t)\D t)$$
Additionally, given the integrability assumption on $\nabla_z h$, we have
$$0=\nabla_z 1=\nabla_z\int h(t,z,a)\D a =\int\nabla_zh(t,z,a)\D a =\int\nabla_z(\log h)(t,z,a) h(t,z,a)\D a .$$
Injecting these identities to the filtering equation, we can now conclude that $h(t,Y_t,a)$ is the unique solution to the filtering equation on $[0,T]$ and is therefore the density of $\tilde \zeta$ conditional on $\cF^Y_t.$
Finally, integrating against this density, we conclude that on its own filtration, $Y$ satisfies
\begin{align}
    \D Y_t&=\int \Sigma\left(\nabla_z (\log h)(t,Y_t,a)+ \gamma(t,Y_t)\right)h(t,Y_t,a)\D a \D t+\D\hat Y_t\notag\\
    &=\Sigma \gamma(t,Y_t)\D t+\D\hat Y_t\label{eq:YunderP}
\end{align}
where $\hat Y_t$ is a $(\P,\cF^Y)$ Brownian motion with covariance matrix $\Sigma$ on $[0,T]$. 
Therefore, the $(\P,\cF^Y)$ distribution of $Y$ is the $\hat \nu$-distribution of $Z$. 

We fix $\tilde \zeta$ belonging to the (full measure) set of differentiability of $\nabla \Gamma$. Thanks to the differentiability of $\nabla \Gamma$, and the integrability of $\tilde \xi$ we have that 
$$\phi(t,z)\to \a\beta'\nabla \Gamma(\tilde \zeta)-\a\Gamma^*(\nabla\Gamma(\tilde \zeta))\mbox{ as }(t,z)\to (T,\tilde \zeta)$$
and 
$$\psi(t,z,\tilde \zeta)\to \a\beta'\nabla \Gamma(\tilde \zeta)-\a\Gamma^*(\nabla\Gamma(\tilde \zeta))\mbox{ as }(t,z)\to (T,\tilde \zeta).$$
Thus, 
\begin{align}\label{eq:hh}
    \frac{h(t,z,\tilde \zeta)}{k(t,z,\tilde \zeta)}\to 1 \mbox{ as }(t,z)\to (T,\tilde \zeta).
\end{align}
We define on $C[0,T]^n$ a measure $\hat \nu^{\tilde \zeta}$ by
$$\left.\frac{\D \hat \nu^{\tilde \zeta}}{\D \hat \nu}\right|_{\cF^Z_t}=h(t,Z_t,\tilde \zeta).$$
Thanks to strong uniqueness of solutions of \eqref{eqy} and \citet[Expression (39.4)]{rogers2000diffusions}, $\hat \nu^{\tilde \zeta}$ is the $(\P,\cF^{\tilde \zeta,Z})$ distribution of $Y.$
Thanks to the convergence \eqref{eq:hh}, we have that $Y_t\to \tilde \zeta$ a.s. as $t\to T. $

{\it Step 3: Defining the risk neutral probability measure.}  Set
$$M_t = \exp\left(-\frac{1}{2}\int_0^t \gamma(s,Y_s)'\Sigma\gamma(t,Y_s)\,
\D s - \int_0^t \gamma(s,Y_s)'\,\D \hat{Y}_s \right)\,$$
where $\hat Y$ is the $(\P,\cF^Y)$ Brownian motion with covariance matrix $\Sigma$ defined by the equality \eqref{eq:YunderP}. $M$ is clearly a $(\P,\cF^Y)$-local martingale.  The pricing rule is $\nabla\Gamma(t,Y_t)$. Thus, the realized wealth $\tw$ of the market makers is as stated in \eqref{wealth3}.  
We want to show that $M$ is a martingale and 
\begin{equation}\label{MT=margutility}
 M_T = \frac{\E^{-\alpha \tw}}{\mye[\E^{-\alpha \tw}]}\,.
\end{equation}
We have
\begin{align*}
 \log M_T &=\log M_t -\frac{1}{2}\int_t^T \gamma(s,Y_s)'\Sigma\gamma(s,Y_s)\,
\D s - \int_t^T \gamma(s,Y_s)'\,\D \hat{Y}_s\\
&= \log M_t+\frac{1}{2}\int_t^T \gamma(s,Y_s)'\Sigma\gamma(s,Y_s)\,
\D s - \int_t^T \gamma(s,Y_s)'\,\D Y_s\,
\end{align*}
and
\begin{align*}
    \mye\left[\frac{M_T}{M_t}|\cF^Y_t\right]&=\mye\left[\E^{\frac{1}{2}\int_t^T \gamma(s,Y_s)'\Sigma\gamma(s,Y_s)\,
\D s - \int_t^T \gamma(s,Y_s)'\,\D Y_s}|\cF^Y_t\right]\\
&=\mye\left[\E^{\frac{1}{2}\int_t^T \gamma(s,Y_s)'\Sigma\gamma(s,Y_s)\,
\D s - \int_t^T \gamma(s,Y_s)'\,\D Y_s}|Y_t\right]
\end{align*}
Since the $(\P,\cF^Y)$ distribution of $Y$ is the $\hat \nu$-distribution of $Z$, to show that $M$ is a martingale, it is sufficient to show that 
$$\mye^{\hat \nu}\left[\E^{\frac{1}{2}\int_t^T \gamma(s,Z_s)'\Sigma\gamma(s,Z_s)\,
\D s - \int_t^T \gamma(s,Z_s)'\,\D Z_s}|Z_t\right]=1.
$$
Recalling $\rho$ defined at \eqref{eq:rho}, we need to show that
$\mye^{ \hat \nu}\left[\E^{-\rho_T+\rho_t}|\cF^Z_t\right]=1.$
Since $\D\hat \nu/\D \nu=\tilde \xi=\E^{\rho_T}$,
we have $$\mye^{ \hat \nu}\left[\E^{-\rho_T+\rho_t}|\cF^Z_t\right]=\frac{\mye^{  \nu}\left[\E^{\rho_T-\rho_T+\rho_t}|\cF^Z_t\right]}{\mye^{  \nu}\left[\E^{\rho_T}|\cF^Z_t\right]}=1\,,$$
and $M$ is a martingale. 
We can define $\Q$ by $\D \Q / \D \P=M_T$. Thanks to the martingality of $M$, by a direct computation the $(\Q,\cF^Y)$-distribution of $Y$ is a $(0,\Sigma)$-Brownian motion, which is the $\nu$-distribution of $Z$. 
Thus, the backward equations \eqref{eq:rho} and 
\begin{align*}
        \tilde \rho_t=\log (\tilde \xi(Y_\cdot))+\frac{1}{2}\int_t^T\tilde \gamma_s'\Sigma \tilde \gamma_s \D s-\int_t^T \tilde \gamma_s'\D Y_s,
\end{align*}
have the same data. 
Thanks to the uniqueness result in \citet[Proposition 5.2]{barrieu2013monotone}, we have that 
$\tilde \rho_t=\rho_t (Y_\cdot)$ and $\tilde \gamma_s=\gamma(s,Y_s)$. Thus, $M_T=\E^{\tilde \rho_T}$ and 
\eqref{MT=margutility} holds and $\Q$ is indeed the risk neutral measure. 

{\it Step 4: $\Q$ dynamics of the price process and equilibrium.}
Because $\nabla \Gamma$ transports the normal $(0,T\Sigma)$ distribution to ${F}$, it follows that the $(\Q,\cF^Y)$--distribution of $\nabla\Gamma(Y_T)$ is ${F}$.  Furthermore, the identity $Y_T=\tilde \zeta$ implies $\nabla\Gamma(Y_T) = \tv$, so the $\Q$--distribution of $\tv$ is ${F}$, as claimed.  Also, the facts that $Y$ is a $(\Q,\cF^Y)$--Brownian motion, $\nabla \Gamma(t,y)$ solves the heat equation, and $\nabla \Gamma (Y_T) = \tv$ imply that the pricing rule \eqref{price1} satisfies the equilibrium condition \eqref{eq2}.  The formula for $\D P$ follows from $P_t = \nabla \Gamma(t,Y_t)$ and the fact that $\nabla \Gamma$ satisfies the heat equation. 
It remains only to show that $\theta$ is an optimal strategy.  The proof of this is the same as for Theorem~\ref{thm:riskneutral}, up to showing that the strategy satisfies the restriction \eqref{restriction}.  First, we note that
$$\mye^\P \int_0^T \tv'\,\D Z_s = \mye^P[\tv' Z_T] = \mye^P[\tv]'\mye^P[Z_T]=0$$
due to the independence of $\tv$ and $Z$, the fact that $\mye^\P[Z_T]=0$, and the fact that $\tv$ has a finite $\P$-mean. 
Set $L_t = \int_0^t H(s,Y_s)'\,\D Z_s$.  It remains only to show that $\mye^{\P}[L_T]=0$, which will follow from showing that $L$ is a $\P$--martingale.  It is a local martingale under $\P$, so, by the dominated convergence theorem, it suffices to show that $\sup_{0 \leq t \leq T} |L_t|$ is $\P$--integrable.   
%and only requires $\mye\left[\int_0^T H(s,Y_s)'\, \D Z_s\right]=0$, $\mye\left[|\tilde v' Z_T|\right]<\infty$, and $Y_T=\tilde \zeta$. Since the latter has already been showed, we first prove that the $\P$-local martingale $t\mapsto\int_0^t H(s,Y_s)'\, \D Z_s$ of the filtration generated by $\tilde v$ and $Z$ is a martingale. By an application of the Dominated convergence theorem, a sufficient condition for this martingality is the integrability of the random variable $\sup_{0\leq t\leq T}\left|\int_0^t H(s,Y_s)'\, \D Z_s\right|$. 
By the Burkholder-Davis-Gundy inequality, the fact that $\D P/\D Q = M_T^{-1}$, and the fact that the $\Q$--distribution of $Y$ is the same as the $\nu$--distribution of $Z$, we have
\begin{align*}
    \mye^{\P} \left[\sup_{0\leq t\leq T}|L_t|\right]&\leq C \mye^{\P} \left[\left(\int_0^T \| H(s,Y_s)\|^2\, \D s\right)^{1/2}\right]\\
    &= C \mye^\Q \left[M_T^{-1}\left(\int_0^T \|H(s,Y_s)\|^2\, \D s\right)^{1/2}\right]\\
    &= C \mye^\nu \left[\xi_T\left(\int_0^T \|H(s,Z_s)\|^2\, \D s\right)^{1/2}\right]
\end{align*}
for a constant $C$.  The Cauchy–Schwarz inequality applied to the last expression gives us
$$\mye^{\P} \left[\sup_{0\leq t\leq T}|L_t|\right] \leq C \mye^\nu \left[\tilde \xi^2\right]^{1/2}\mye^\nu \left[\int_0^T \|H(s,Z_s)\|^2\, \D s\right]^{1/2}\,.$$
%\begin{align*}
% \mye \left[\sup_{0\leq t\leq T}\left|\int_0^t H(s,Y_s)'\, \D Z_s\right|\right]
%&\leq C \mye^\nu \left[\tilde \xi \left(\int_0^T \|H(s,Z_s)\|^2\, \D s\right)^{1/2}\right]\\
%&\leq C \mye^\nu \left[\tilde \xi^2\right]^{1/2}\mye^\nu \left[\int_0^T %\|H(s,Z_s)\|^2\, \D s\right]^{1/2}.
%\end{align*}
$ \mye^\nu \left[\tilde \xi^2\right]$ is finite by assumption. Additionally, we have
$$\mye^\nu\int_0^T \|H(s,Z_s)\|^2\, \D s = \mye^\Q \int_0^T \|H(s,Y_s)\|^2\, \D s = \mye^\Q \int_0^T \left\|\mye^\Q[\tv \mid \cF^Y_s]\right\|^2\, \D s\,,$$
which is finite by Jensen's inequality and because $\tv$ has finite second moments under the risk-neutral distribution.
%\begin{align*}
%    \mye^\nu \left[\int_0^T \|H(s,Z_s)\|^2\, \D s\right]&\leq2T\mye^\nu \left[ \|\nabla\Gamma(Z_T)\|^2\right]
%\end{align*}
%where we use the fact that $H(s,Z_s)=\mye^\nu[\nabla\Gamma(Z_T)|\cF^Z_s]$ and Jensen's inequality. The $\nu$-distribution of $\nabla \Gamma(Z_T)$ is $F$ which admits a second moment. To prove $\mye\left[|\tilde v' Z_T|\right]<\infty$, we use the independence of $\tilde v$ and $Z_T$, to obtain
%$\mye\left[|\tilde v' Z_T|\right]\leq C \mye\left[\|\tilde v\|\right]$ for a constant $C$ depending only on $T\Sigma$. Note also that 
%$$\mye\left[\|\tilde v\|\right]=\mye^\nu\left[\tilde \xi\|\nabla\Gamma (Z_T)\|\right]\leq \mye^\nu\left[\tilde \xi^2\right]^{1/2}\mye^\nu\left[\|\nabla\Gamma (Z_T)\|^2\right]^{1/2}<\infty.$$
%Thus, we obtain the optimality of $\theta.$

\end{proof}

%%%%%%%%%%%%%%%%%%%%%%%%%%%%%%%%%%%%%%%%%%%%
%
% Proof of Theorem 4.2 (Risk Averse Equilibrium with Assumptions on Fundamentals)
%
%%%%%%%%%%%%%%%%%%%%%%%%%%%%%%%%%%%%%%%%%%%%

\begin{proof}[Proof of Theorem~\ref{thm:riskaverse2}]
{\it Case (i):}
If the support of the distribution is bounded, then using the integration by parts formula of Malliavin calculus we obtain
\begin{align}\label{boundtr2}
    \tr\left(\Sigma \nabla^2 \Gamma(t,z)\right)\leq \mye^\nu\left[\frac{|Z_T-Z_t|}{T-t}|\nabla \Gamma(t,z+Z_T-Z_t)|\right]\leq \frac{C}{\sqrt{T-t}}
\end{align}
for some constant $C$ depending only on the bound of the price. Thus, the integral 
$$\int_0^T\tr\left(\Sigma \nabla^2 \Gamma(t,Z_t)\right)\D t$$
is uniformly bounded. 
Since $-\a \Gamma^*(\nabla \Gamma)$ is bounded from above, we have the integrability of $\tilde \xi^2$ and $\tilde \chi^2$.
We can also show that $\Gamma$ is Lipschitz continuous. Thus, $-\a  \Gamma^*(\nabla\Gamma(\cdot))$ is also bounded from below by a function with linear growth. Thus, thanks to the existence of exponential moments of the Gaussian distribution $\tilde \xi^{-1}$ is also integrable. 

For fixed $t,z,b$, using the distribution of the conditioned Brownian motion, we can write 
\begin{align*}
    &\mye^\nu\left[\left.\exp\left\{ \alpha \tr\left(\Sigma \int_t^T \nabla^2 \Gamma(u,Z_u)\,\D u  \right)\right\}\,\right|\, Z_t=z, Z_T=b\right]\\
   &=\mye^\nu\left[\exp\left\{ \alpha \tr\left(\Sigma \int_t^T \nabla^2 \Gamma\left(u,\frac{T-u}{T-t}z+\frac{u-t}{T-t}(b- Z_T+ Z_t)+ Z_u+ Z_t\right)\,\D u  \right)\right\}\right].
\end{align*}
Thanks to the boundedness of $\int_0^T\tr\left(\Sigma \nabla^2 \Gamma(t,Z_t)\right)\D t$, we can differentiate this expression under the integral in $z$ to show that the partial derivatives of $\psi$ admit the following representations
\begin{align*}
\frac{\partial \E^{\psi(t,z,b)} }{\partial z_i} =& \E^{-\alpha(b-\beta)'\nabla\Gamma(b)+\alpha \Gamma(b)} \mye^\nu\left[\exp\left\{ \alpha \tr\left(\Sigma \int_t^T \nabla^2 \Gamma(u,Z_u)\,\D u  \right)\right\}\right.\\ &\times\left.\left.\int_t^T\frac{T-u}{T-t}\Gamma_i(u,Z_u)\D u \,\right|\, Z_t=z, Z_T=b\right]\\
\frac{\partial^2\E^{\psi(t,z,b)}}{\partial z_i\partial z_j}  =&\E^{-\alpha(b-\beta)'\nabla\Gamma(b)+\alpha \Gamma(b)} \mye^\nu\left[ \exp\left\{ \alpha \tr\left(\Sigma \int_t^T \nabla^2 \Gamma(u,Z_u)\,\D u  \right)\right\}\right.\\
&\times\left.\left.\left\{\Psi_{i,j}+\int_t^T\frac{(T-u)^2}{(T-t)^2}\Gamma_{i,j}(u,Z_u)\D u \right\}\,\right|\, Z_t=z, Z_T=b\right]. 
\end{align*}
where $\Psi_{i,j}=\int_t^T\frac{T-u}{T-t}\Gamma_i(u,Z_u)\D u\int_t^T\frac{T-u}{T-t}\Gamma_j(u,Z_u)\D u$ and $\Gamma_i,\Gamma_j$ depends on $\nabla^3 \Gamma$ and $\Gamma_{i,j}$ depends on $\nabla^4 \Gamma$. Similarly to \eqref{boundtr2}, each of these functions admits the bounds 
$$|\nabla^3 \Gamma(u,z)|\leq \frac{C}{T-u}\mbox{ and }|\nabla^4 \Gamma(u,z)|\leq \frac{C}{(T-u)^{3/2}}$$
and we easily obtain that $\psi$ is twice continuously differentiable in $z$ on $[0,T)$ with bounded first and second order derivative on $[0,T-\e]$ for all $\e>0$. By integrating $\psi$ in $b$ this is also true for $\phi$. 
We can also obtain the smoothness of these functions in $t$. 
By Ito's formula, we can directly establish that 
$$\log \xi_t =\phi(t,Z_t)+\int_0^t \a \tr\left(\Sigma \nabla^2\Gamma(u,Z_u)\right)\D u$$ and 
$\gamma(t,z)=\nabla \phi(t,z)$ for $t\in [0,T)$. 

The expression \eqref{eq:h} is a direct consequence of the $\nu$ distribution of $Z$ and the definition $\D \hat \nu/\D \nu=\tilde \xi$.
Given \eqref{eq:h}, we easily obtain that $h$ is twice continuously differentiable in $z$ and continuously differentiable in $t$ for $t<T$ and $\nabla_z \psi$ has linear growth in $z$ for $t<T-\epsilon$ for all $\epsilon>0$. Thus, \eqref{eqy} admits a unique strong solution on $[0,T)$. The uniqueness of the solution on $[0,T]$ is a consequence of the fact that 
$(\P,\cF^Y)$ distribution of $Y$ is the $\hat \nu$ distribution of $Z$. Thus, $Y$ is continuous and therefore can be extended to $[0,T]$ uniquely. 
The expression \eqref{informed2} can now be deduced from \eqref{eqy}.

{\it Case (ii):}  The inequality $\nabla U(v)\geq \kappa I_n$ implies that there exists a constant $C>0$ so that 
$$\E^{-U(v)}\leq C\E^{-\frac{\kappa\|v\|^2}{4}}.$$ Thus, $F$ has finite second moments. The version of Caffarelli's contraction theorem in \citet[Corollary 6.1]{kolesnikov2013sobolev} implies that $\nabla^2 \Gamma(y)$ (and therefore $\nabla^2 \Gamma(t,y)$) is uniformly bounded. This point also means that $\beta'\nabla \Gamma(z)$ has at most linear growth in $z.$ Thus, thanks to the finiteness of exponential moments of the absolute value of the Gaussian distribution, the conclusion of the previous case holds in this case. 

\end{proof}

%%%%%%%%%%%%%%%%%%%%%%%%%%%%%%%%%%%%%%%%%%%%
%
% Proof of Theorem 5.1
%
%%%%%%%%%%%%%%%%%%%%%%%%%%%%%%%%%%%%%%%%%%%%

\begin{proof}[Proof of Theorem~\ref{thm:normal}] 

Adopt the definitions at the end of Section~\ref{s:normal}.  
Let $F$ denote the normal $(m,S)$ distribution function.  Then, condition (ii) of Theorem~\ref{thm:riskaverse2} holds, so the assumptions of Theorem~\ref{thm:riskaverse} hold.  
The Brenier potential for transporting $G$ to $F$ is $\Gamma(y) = a+ (1/2)y'\Lambda y+ m'y$, where the constant $a$ is determined by normalization, and the transport map is $\nabla \Gamma(y) = m + \Lambda y$.  To see that this transports $G$ to $F$, it suffices to observe that $T\Lambda\Sigma\Lambda=S$.  Applying the definition \eqref{price1} to this transport map implies directly that the equilibrium pricing rule is $P_t = m + \Lambda Y_t$.  To compute the constant $a$ in the potential, we observe that
\begin{equation}\label{normalconstant}
    \mye^G\left[\frac{1}{2}\tilde z'\Lambda \tilde z + m'\tilde z\right] = \frac{1}{2}\tr(T\Sigma\Lambda) = \frac{1}{2}\tr(D^{1/2})\,.
\end{equation}
Thus, $\Gamma(y) = (1/2)y'\Lambda y+ m'y - (1/2)\tr(D^{1/2})$.

We can explicitly compute $\psi$ as
\begin{align*}
    \psi(t,z,a) = \alpha\beta'\nabla\Gamma(a)-\alpha \Gamma^*(\nabla \Gamma(a))+\a(T-t)\tr\left(\Sigma \Lambda  \right)\,.
\end{align*}
This does not depend on $z$, so the formula \eqref{informed2} implies \eqref{informedgaussian}. 
Define
$$C_t = \frac{1}{\alpha T}\Sigma^{-1/2}V'D_tV\Sigma^{-1/2}\,.$$
%$$C_t = [\Lambda^{-1} + \alpha(T-t)\Sigma]^{-1}=\Sigma^{-1/2}[(\Sigma^{1/2}\Lambda\Sigma^{1/2})^{-1} + \alpha(T-t)I_n]^{-1}\Sigma^{-1/2}$$
and
\begin{align*}
    E_t&= \b' m+\left(\frac{T}{2}-t\right)\tr\left(\Sigma \Lambda\right)+\int_t^T  \left\{\frac{\a}{2}\b'C_s'\Sigma C_s\b-\frac{1 }{2}\tr\left(\Sigma C_s\right) \right\}\,\D s\,.
\end{align*}
Then, we can compute that 
$$\phi(t,z)=-\frac{\a}{2}z'C_t z+\a \b'C_tz+\a E_t\,.$$ 
Therefore,
\begin{equation}\label{newdY3}
\D Y_t= -\a \Sigma C_t (Y_t-\b)\, \D t +\D \hat Y_t\,,
\end{equation}
where $\hat Y$ is a $(0,\Sigma)$--Brownian motion on $(\P,\cF^Y)$.  This verifies \eqref{corY}.  Furthermore, \eqref{newdY3} implies that
$$\D V\Sigma^{-1/2}Y_t = -\frac{1}{T}D_t (V\Sigma^{-1/2}Y_t - V\Sigma^{-1/2}\beta)\,\D t + V\Sigma^{-1/2}\,\D \hat Y_t\,.
$$
Note that, due to the orthogonality of $V$, $V\Sigma^{-1/2}\,\D \hat Y_t$ is a  $(\P,\cF^Y)$-standard Brownian motion. Thus, each component $y_{it}$ of the vector $V\Sigma^{-1/2}(Y_t-\beta)$ solves a stochastic differential equation independent from the others:
$$\D y_{it} = -\frac{ d_{it}}{T} y_{it}\,\D t + \D w_{it}\,,
$$
where the $w_{i}$ are independent standard Brownian motions.  The solutions of these equations are
$$y_{it} = \frac{T + \alpha(T-t)\sqrt{d_i}}{T + \alpha T \sqrt{d_i} }\left(y_{i0} + \int_0^t \frac{T + \alpha T \sqrt{d_i}}{T + \alpha(T-s)\sqrt{d_i}}\,\D w_{is} \right)\,.
$$
Let $\tilde D_t$ denote the diagonal matrix with $i$th element equal to
$$\frac{T + \alpha(T-t)\sqrt{d_i}}{T + \alpha T \sqrt{d_i} }\,.$$
It follows that $V\Sigma^{-1/2}(Y_t-\beta)$ is normally distributed with mean  equal to 
$$\tilde D_t V\Sigma^{-1/2}(Y_0-\beta) = - \tilde D_t V\Sigma^{-1/2}\beta\,.$$
Hence, $Y_t$ is normally distributed with mean  equal to
\begin{equation*}
    \beta -\Sigma^{1/2}V'\tilde D_t V\Sigma^{-1/2}\beta = \Sigma^{1/2}V'(I-\tilde D_t)V\Sigma^{-1/2}\beta\,.
\end{equation*}
Therefore, $P_t = m + \Lambda Y_t$ is normally distributed with mean equal to
\begin{equation}\label{meanPt}
    m + \frac{1}{T}\Sigma^{-1/2}V'D^{1/2}(I-\tilde D_t)V\Sigma^{-1/2}\beta
\end{equation}
From the definition of the $d_i$ and $\tilde D_t$, we obtain 
$D^{1/2}(I-\tilde D_T) = \alpha \hat D$.  Substituting $m=\hat m - \alpha \hat S \beta$ in \eqref{meanPt}, we see that the mean of $\tv = P_T $ is
$$\hat m - \alpha \hat S \beta + \frac{\alpha}{T}\Sigma^{-1/2}V'\hat D V \Sigma^{-1/2}\beta = \hat m\,.$$  
To verify that the physical distribution $\hat F$ defined in Theorem~\ref{thm:riskaverse} is normal $(\hat m, \hat S)$, it remains to verify that the covariance matrix of $P_T$ is $\hat S$.  

Thanks to the independence of $w^i$ and $w^j$ for $i\neq j$, the covariance matrix of $(y^1_t,\ldots, y^n_t)$
is diagonal with entries 
$$\int_0^t \left[\frac{T+\alpha(T-t)\sqrt{d_i}}{T+\alpha (T-s)\sqrt{d_i}}\right]^2\,\D s = t \left[\frac{T+\alpha(T-t)\sqrt{d_1}}{T+\alpha T \sqrt{d_i}}\right]\,.$$
At $t=T$, this equals $T/(1+\alpha \sqrt{d_i}) = T \hat{d_i}/d_i$.  Thus, the covariance matrix of $(y^1_T,\ldots, y^n_T)$ is $T \hat D D^{-1}$.  It follows that the covariance matrix of $Y_T$ is
$\Sigma^{1/2}V'\hat D D^{-1} V \Sigma^{1/2}$ and therefore that the covariance matrix of $P_T = m + \Lambda Y_T$ is 
$$T\Lambda \Sigma^{1/2}V'\hat D D^{-1} V \Sigma^{1/2} \Lambda = \frac{1}{T}\Sigma^{-1/2}V' \hat D V \Sigma^{-1/2} = \hat S\,.$$

Now, set $A_t = \alpha V'D_tV/T = \alpha \Sigma^{1/2}C_t \Sigma^{1/2}$. Then, $\alpha \Sigma C_t = \Sigma^{1/2}A_t\Sigma^{-1/2}$, so \eqref{newdY3} implies \eqref{corY}.  Furthermore, 
\begin{align*}
    \D P_t = \Lambda\,\D Y_t &= - \Lambda \Sigma^{1/2} A_{ t} \Sigma^{-1/2}(Y_t-\beta)\,\D t + \Lambda\,\D \hat Y_t\\
    &= - \Lambda \Sigma^{1/2} A_{ t} \Sigma^{-1/2}\Lambda^{-1}(P_t-m-\Lambda\beta)\,\D t + \Lambda\,\D \hat Y_t\\
    &= -\Sigma^{-1/2}A_t \Sigma^{1/2}(P_t-m-\Lambda\beta)\,\D t + \Lambda\,\D \hat Y_t\,,
\end{align*}
which verifies \eqref{corP}.
The monotonicity in $\alpha$ of $S$, $\Lambda$, and $A_t$ follows from the fact that the elements of the diagonal matrices $D$, $D^{1/2}$ and $\alpha D_t$ are increasing in $\alpha$.

The expected loss of noise traders is the sum of the covariances of the prices $P^i$ with the noise trades $Z^i$, which, as in \eqref{normalconstant}, is $\mye^G[\tilde z'\Lambda \tilde z] = \tr(T\Sigma\Lambda) = \tr(D^{1/2})$.   As in Theorem~\ref{thm:riskneutral}, the expected profit of the informed trader conditional on $\tv$ is $\Gamma^*(\tv)$.  Given the expression for $\Gamma$, we have 
\begin{align*}
   \Gamma^*(v)=\sup_{y\in \R^n} \{y'v-\Gamma(y)\} &=\frac{1}{2}(v-m)'\Lambda^{-1} (v-m)+\frac{1}{2}\tr(T\Sigma\Lambda) \\ 
   &=\frac{1}{2}(v-\hat m - \alpha \hat S \beta)'\Lambda^{-1} (v-\hat m - \alpha \hat S \beta)+\frac{1}{2}\tr(T\Sigma\Lambda)\,,
\end{align*}
which verifies \eqref{cor:expprofit}. The expectation of $\Gamma^*(\tv)$ under the physical distribution is
\begin{multline}\label{expprofitagain}
\frac{1}{2} \mye^{\hat F}[(\tv-\hat m)'\Lambda^{-1}(\tv - \hat m)] +
\frac{\alpha^2}{2}\beta'\hat S \Lambda^{-1} \hat S \beta + \frac{1}{2}\tr(T\Sigma\Lambda) \\= \frac{1}{2}\tr(\Lambda^{-1}\hat S) +
\frac{\alpha^2}{2}\beta'\hat S \Lambda^{-1} \hat S \beta + \frac{1}{2}\tr(T\Sigma\Lambda)\,.
\end{multline}
A direct calculation gives
$\Lambda^{-1}\hat S = \Sigma^{1/2}V'D^{-1/2}\hat D V \Sigma^{-1/2}$,
so $\tr(\Lambda^{-1}\hat S) = \tr(D^{-1/2}\hat D)$.  From \eqref{DDhat}, we have $D^{-1/2}\hat D = D^{1/2}-\alpha \hat D$, so \eqref{expprofitagain} equals
$$\frac{\alpha^2}{2}\beta'\hat S \Lambda^{-1} \hat S \beta + \tr(D^{1/2}) - \frac{\alpha}{2}\tr(\hat D)\,.$$
The formula \eqref{cor:expprofit2} follows from observing that $\tr(D^{1/2}) = \tr(T\Sigma\Lambda)$ and $\tr(\hat D) = \tr(T\Sigma \hat S)$. The formula for the market makers' expected gains/losses follows from the expected loss of the noise traders and the expected gain of the informed trader. 
The term $(\alpha^2/2)\beta'\hat S \Lambda^{-1} \hat S \beta$ in  \eqref{cor:expprofit2} is increasing in $\alpha$, because 
$\alpha^2 \Lambda^{-1} = \Sigma^{1/2}V'(\alpha^2D^{-1/2})V\Sigma^{1/2}$,
and the elements of the diagonal matrix $\alpha^2D^{-1/2}$ are
$$\frac{2\alpha}{ \hat d_i + \sqrt{ \hat d_i^2 + 4 \hat d_i/\a^2}}\,,$$
which are increasing in $\alpha$.
The remaining part of \eqref{cor:expprofit2} is $\tr(D^{1/2}-(\alpha/2)\hat D)$, which is also increasing in $\alpha$, because the elements of the diagonal matrix $D^{1/2}-(\alpha/2)\hat D$ are 
$\sqrt{\alpha^2 \hat d_i^2 + 4 \hat d_i}/2$,
which are increasing in $\alpha$.  
\end{proof}

\end{appendix}

\bibliographystyle{elsarticle-harv}
\bibliography{zotero-export2}

\newpage
\begin{appendix}
\setcounter{page}{1}
\section*{ONLINE APPENDIX}
The online appendix consists of three parts: proofs of lemmas, a description of numerical methods used for the options model, and an example of a single-asset model with risk-averse market makers that illustrates the role of the market makers' endowments.

\subsection{Proofs of Lemmas}

%%%%%%%%%%%%%%%%%%%%%%%%%%%%%%%%%%%%%%%%%%%%%%%%%
%
% Proof of Lemma A.1 
%
%%%%%%%%%%%%%%%%%%%%%%%%%%%%%%%%%%%%%%%%%%%%%%%%%

\begin{proof}[Proof of Lemma A.1]
By the definition of $\Gamma$, 
$\nabla \Gamma(Z_T)$ has the distribution of $\tv$.  Hence, $\nabla \Gamma(Z_T)$ is integrable.  By conditioning, it follows that, for each $t \in (0,T)$, $\nabla \Gamma(t,y+Z_T-Z_t)$ is integrable for Lebesgue almost all $y$.  We will show that this extends to all $y$ in a moment.  We turn now to the integrability of $\Gamma(Z_T)$.

By the convexity of $\Gamma$, we have
$$\Gamma(0) + \nabla \Gamma(0)'Z_T \le \Gamma(Z_T) \le \Gamma(0) + \nabla \Gamma(Z_T)'Z_T\,.$$
Because all moments of $Z_T$ exist and $\|\nabla \Gamma(Z_T)\|^2$ is integrable, the H\"{o}lder inequality implies that $\nabla \Gamma(Z_T)'Z_T$ is integrable.  Thus, $\Gamma(Z_T)$ is bounded above and below by integrable random variables and hence is itself integrable.  By conditioning, it follows that, for all $t>0$, $\Gamma(y+Z_T-Z_t)$ is integrable for Lebesgue almost all $y$.

We now want to show that $\nabla \Gamma(y+Z_T-Z_t)$ and $\Gamma(y+Z_T-Z_t)$ are integrable for all $y$.  Fix $t$, and let $k_y$ denote the normal $(y,(T-t)\Sigma)$ density function.  We need to show that
\begin{subequations}\label{toshow}
\begin{equation}\label{toshow1}
\int \left\| \nabla \Gamma(z) \right\| k_y(z)\,\D z < \infty
\end{equation}
and 
\begin{equation}\label{toshow2}
\int \left\| \Gamma(z) \right\| k_y(z)\,\D z < \infty
\end{equation}
\end{subequations}
for all $y$, knowing that these conditions hold for Lebesgue almost all $y$.  Fix an arbitrary $y^*$ and set
$B = \{y : (y-y^*)'\Sigma (y-y^*) \le 1\}$.
This is the unit ball centered at $y^*$ in the norm defined by $\Sigma$.
There exists $y$ in the ball $B$ for which \eqref{toshow} holds, and the ratio $k_{y^*}(z)/k_y(z)$ is bounded over $z\in B$, so \eqref{toshow} holds for $y^*$ when we restrict the integrals to the ball $B$.  Now, let $y_1, \ldots, y_n$ and $\epsilon<1$ be such that the $n$ $\epsilon$--balls 
$\left\{y : (y-y_i)'\Sigma (y-y_i) < \epsilon^2\right\}$
cover the boundary of $B$.  Because \eqref{toshow} holds for Lebesgue almost all $y$, we can choose each $y_i$ so that \eqref{toshow} holds for it.  For each $z$ not in the unit ball, there exists by the triangle inequality some $i$ such that $y_i$ is closer to $z$ than $y^*$ is to $z$ in the norm defined by $\Sigma$ [the triangle inequality applies as follows: let $\hat{y}$ be the point on the boundary of $B$ that is on the line segment between $z$ and $y^*$ and suppose $\hat{y}$ is in the $i$th $\epsilon$--ball; then, the distance between $z$ and $y_i$ is less than $\epsilon$ plus the distance between $z$ and $\hat{y}$, which is less than 1 plus the distance between $z$ and $\hat{y}$, which is the distance between $z$ and $y^*$].
Hence, for all $z \notin B$, 
$k_{y^*}(z) < \max\,\{ k_{y_i}(z) \mid i=1, \ldots, n\}$.
Because \eqref{toshow} holds for each $i$, this implies that \eqref{toshow} holds for $y^*$ when we restrict the integrals to the complement of $B$.  Thus, it holds for the integrals over $\R^n$.
\end{proof}

%%%%%%%%%%%%%%%%%%%%%%%%%%%%%%%%%%%%%%%%%%%%%%%%%
%
% Proof of Lemma A.2 
%
%%%%%%%%%%%%%%%%%%%%%%%%%%%%%%%%%%%%%%%%%%%%%%%%%

\begin{proof}[Proof of Lemma A.2]
From \citet[Theorem 5.10]{Villani:2009:OptimalTransportOld}, there exists $\hat \Phi\in \Pi(L,F)$, satisfying 
\begin{align*}
&\int v'y\hat \Phi(dy,dv)=\sup_{\Phi\in \Pi(L,F)} \int v'y\Phi(dy,dv)\\
&=\inf_{\psi\in L^1(L)}\int \psi(y)\,\D L(y)+\int \psi^*(v)\,\D F(v)=\int \Gamma(y) \, \D L(y)+\int \Gamma^*(v)\,\D F(v).
\end{align*}
There exists a disintegration of the measure $\hat \Phi$ on $F$, meaning there exists a family of probability measures $\hat \phi$ such that 
 $\hat \Phi(\D y,\D v)=\hat \phi(\D y,v) \D F(v)$, and
  $v\in \R^n\mapsto \hat \phi(A,v)$ is measurable for any Borel measurable subset $A$ of $\R^n$.
Thanks to \citet[Theorem 5.30]{Villani:2009:OptimalTransportOld}, we also know that $\hat \Phi$ is supported on 
$\{(y,\nabla \Gamma(y)):y\in \R^n\}$. Thus, we also have that for $F$ almost every $v\in \R^n$, the measure $\hat \phi(\D y,v)$ is supported on $(\nabla \Gamma)^{-1}(v)=\{y\in \R^n : \nabla \Gamma(y)=v\}$.

If $F$ is absolutely continuous, then $\nabla \Gamma$ is bijective,  and its inverse transports $F$ to $L$ 
\citep[][Proposition 3.1]{Brenier:1991:PolarFactorizationMonotone}. Therefore, we can take  $f(u,v)=(\nabla \Gamma)^{-1}( v)$.
If $\nabla \Gamma$ fails to be bijective, let $U$ be any random variable with distribution $G$. 
For each $v\in \R^n$, we use Brenier's theorem for the transport of $G$ to the probability measure $\hat \phi(\D y,v)$ to obtain the existence of a function $f_v:\R^n\mapsto \R^n$ so that $f_v(U)$ has distribution $\hat \phi(\D y,v)$. Set $f(u,v) = f_v(u)$.  Since $\hat \phi(\D y,v)$ is supported in $(\nabla \Gamma)^{-1}(v)$ we have that $\mu\left(\nabla \Gamma(f(u,v)=v\right)=1$. Moreover, for each Borel $A$, 
\begin{multline*}
G \otimes F (\{(u,v) \mid f(u,v) \in A\}) = \int G(f_v(U) \in A) \,\D F(v) = \int \hat \phi (A,v) \,\D F(v) \\ = \int \hat \Phi(A,\D v) = \int_A \D L\,.
\end{multline*}
The last equality follows from the fact that the first marginal of $\hat \Phi$ is $L$.
\end{proof}

%%%%%%%%%%%%%%%%%%%%%%%%%%%%%%%%%%%%%%%%%%%%%%%%%
%
% Proof of Lemma 3.1 (dGamma)
%
%%%%%%%%%%%%%%%%%%%%%%%%%%%%%%%%%%%%%%%%%%%%%%%%%

\begin{proof}[Proof of Lemma~\ref{lem_dGamma}]
Let $z \mapsto k(t,y,z)$ be the normal $(y,(T-t)\Sigma)$ density function.  The definition of $\Gamma(t,y)$ is
$\Gamma(t,y) = \int \Gamma(z) k(t,y,z)\,\D z$.
Derivatives of $\Gamma(\cdot,\cdot)$ of all order exist and are equal to integrals of $\Gamma(z)$ multiplied by the corresponding derivatives of $k$.  See, for example, \citet[Problem 4.3.1]{Karatzas-Shreve}.  Thus, $\Gamma(t,y)$ has the smoothness claimed.
An equivalent definition of $\Gamma(t,y)$ is
$$\Gamma(t,y) = \int \Gamma(y+z)k(t,0,z)\,\D z\,.$$
We wish to show that we can pass differentiation with respect to $y$ through the integral in this formulation, producing
$$\nabla \Gamma (t,y) = \int \nabla \Gamma(y+z)k(t,0,z)\,\D z \equiv \mye[\nabla \Gamma(y+Z_T-Z_t)] \equiv H(t,y) \,.$$
This holds if, for each $y$, there exists $\epsilon>0$ such that the family of functions
$$\{z \mapsto \nabla \Gamma(y'+z)k(t,0,z) \mid \|y'-y\|<\epsilon\}$$
is uniformly integrable.  Noting that $k(t,0,z) = k(t,y',y'+z)$ and making the change of variables $u=y'+z$, we see that this is equivalent to the family of functions
$$\{u \mapsto \nabla \Gamma(u)k(t,y',u) \mid \|y-y'\|<\epsilon\}$$
being uniformly integrable.  This can be established following reasoning similar to that used in the proof of Lemma~A.1, first bounding $k(t,y',u)$ by a multiple of $k(t,y,u)$ for $u$ in a ball around $y$ and then bounding $k(t,y',u)$ by $\max_{i=1,\ldots,n} k(t,y_i,u)$ for a finite collection $y_1,\ldots,y_n$, for $u$ outside the ball.  
\end{proof}

%%%%%%%%%%%%%%%%%%%%%%%%%%%%%%%%%%%%%%%%%%%%%%%%%
%
% Proof of Lemma 4.1 
%
%%%%%%%%%%%%%%%%%%%%%%%%%%%%%%%%%%%%%%%%%%%%%%%%%

\begin{proof}[Proof of Lemma~\ref{gmarkov}]
By the martingale representation theorem and \citet[Proposition 5.1 (iii)]{barrieu2013monotone}, the integrability of $\tilde \xi$ and $\xi^{-1}$ implies that 
there exists a couple of processes $(\rho_t,\gamma_t)$ adapted to the filtration generated by $Z$ so that $\int_0^t \gamma_s'\,\D Z_s$ is a  $\nu$--martingale with $\mye^\nu[\int_0^T |\gamma_s|^2\D s]<\infty$,
\begin{align}\label{eq:rho}
    \rho_t=\log (\tilde \xi)+\frac{1}{2}\int_t^T\gamma_s'\Sigma \gamma_s \D s-\int_t^T \gamma_s'\D Z_s,
\end{align}
and $\xi_t = \E^{\rho_t}$ is a uniformly integrable $\nu$--martingale with $\D \xi/\xi = \gamma'\,\D Z$.
By Girsanov's theorem,
$
\D \hat Z := \D Z - \Sigma \gamma\,\D t
$
defines a Brownian motion under $\hat\nu$.  

It remains to show that $\gamma_t$ depends on the path of $Z$ only via $Z_t$, so, with an abuse of notation, $\gamma_t = \gamma(t,Z_t)$ for a measurable function $\gamma$.
Define the optimal control problem under weak formulation on $[0,{T/2}]$, 
\begin{align}\label{eq:optphi}
    V(t,z)=\sup_{(x)} \mye^{ \nu_{x}}\left[\phi({T/2},Z_{{T/2}})+\int_t^{{T/2}}\a\tr\left(\Sigma \nabla^2\Gamma(s,Z_s)\right)-\frac{1}{2}x_s'\Sigma x_s \D s\right]
\end{align}
where $x=(x_s)$ is an $\cF^Z$-progressively measurable process with $\mye^\nu\int_t^{{T/2}}\|x_s\|^2\D s<\infty$ and under $\nu_x$ the process $Z$ satisfy for $s\in [t,{T/2}]$, $ Z_s=z+\int_t^s \Sigma x_r \D r+ Z^{x}_s-Z^{x}_t$ for some $(0,\Sigma)$-Brownian motion $Z^{x}$ under $\nu_{x}$. 
The definition of $\xi$ and the definition \eqref{defphi} of $\phi$ imply
$$\xi_t = \exp\left(\phi(t,Z_t) - \phi(0,0) + \alpha \int_0^t \tr(\Sigma \nabla^2\Gamma(s,Z_s))\,\D s \right)\,.$$
In combination with $\D \xi/\xi = \gamma'\,\D Z$, this implies
$$\gamma_t'\,\D Z_t = \frac{\D \xi_t}{\xi_t} = \D \phi(t,Z_t) + \alpha \tr(\Sigma \nabla^2\Gamma(t,Z_t))\,\D t + \frac{1}{2}\gamma_t'\Sigma \gamma_t\,\D t\,.$$
Therefore, for $t<T/2$,
\begin{align*}
    \phi(t,Z_t)=\phi({T/2},Z_{{T/2}})+\int_t^{{T/2}}\left\{\a\tr\left(\Sigma \nabla^2\Gamma(s,Z_s)\right) +\frac{1}{2}\gamma_s'\Sigma \gamma_s \right\}\,\D s-\int_t^{{T/2}} \gamma_s'\,\D Z_s,
\end{align*}
Since 
$$\frac{1}{2}\gamma_s'\Sigma \gamma_s= \sup_{x\in \R^n} \left\{x' \Sigma \gamma_s-\frac{1}{2} x'\Sigma x \right\}$$
with the supremum being achieved at $\gamma_s$, 
we have for all $(x_s)$ as above, 
\begin{align*}
    \phi(t,Z_t)&=\phi({{T/2}},Z_{{T/2}})+\int_t^{{T/2}}\left\{\a\tr\left(\Sigma \nabla^2\Gamma(s,Z_s)\right) +\frac{1}{2}\gamma_s'\Sigma \gamma_s \right\}\,\D s\\
    &-\int_t^{{T/2}} \gamma_s'(\D Z_s-\Sigma \gamma_s \D s)\\
&\geq\phi({{T/2}},Z_{{T/2}})+\int_t^{{T/2}}\left\{\a\tr\left(\Sigma \nabla^2\Gamma(s,Z_s)\right) +\frac{1}{2}x_s'\Sigma x_s \right\}\,\D s\\
    &-\int_t^{{T/2}} \gamma_s'(\D Z_s-\Sigma x_s \D s)\,.
\end{align*}
Since under $\hat \nu$, $\D Z_t-\Sigma\gamma_t \D t$ is a $(0,\Sigma)$-Brownian motion, we have
\begin{align*}
    \phi(t,z)&= \mye^{\hat \nu}\left[\phi({T/2},Z_{{T/2}})+\int_t^{{T/2}}\a\tr\left(\Sigma \nabla^2\Gamma(s,Z_s)\right)-\frac{1}{2}\gamma_s'\Sigma \gamma_s \D s|Z_t=z\right]\\
    &=\sup_{(x)}\mye^{ \nu_x}\left[\phi({T/2},Z_{{T/2}})+\int_t^{{T/2}}\a\tr\left(\Sigma \nabla^2\Gamma(s,Z_s)\right)-\frac{1}{2}x_s'\Sigma x_s \D s\right]. 
\end{align*}
Thus, $\phi(t,z)=V(t,z)$. 
Thanks to the continuity of $\phi$, the nonnegativity of $\tr\left(\Sigma \nabla^2\Gamma\right)$, and the convexity of $x\mapsto x'\Sigma x$, the optimization problem \eqref{eq:optphi} satisfies all the assumption of \citet[Theorem 5.15]{haussmann1990existence}. Thus, the optimal control $\gamma_t$ can be written as feedback control $\gamma(t,Z_t)$ on $[0,{T/2}]$. We can repeat the procedure on the intervals $[(2^n-1)T/2^n,(2^{n+1}-1)T/2^{n+1}]$ to obtain $\gamma(t,z)$ on $(0,T)$. 

\end{proof}

%%%%%%%%%%%%%%%%%%%%%%%%%%%%%%%%%%%%%%%%%%%%
%
\subsection{Numerical Methods for Section~\ref{s:options}}\label{app:numerics}
%
%%%%%%%%%%%%%%%%%%%%%%%%%%%%%%%%%%%%%%%%%%%%

To solve the risk-neutral model, we need to compute the transport map $\nabla \Gamma (y)$ and the pricing rule $H(t,y) = \nabla \Gamma (t,y)$, which is the conditional expectation of $\nabla \Gamma (Z_T)$. In the risk-averse case, we also need to compute the vector of prices of risk $\nabla \phi(t,y)$. Due to the fact that the underlying space for order imbalances is two dimensional, our methods are based on discretization of this space. Higher dimensional problems would have required more sophisticated tools.

\subsubsection{Transport Map and Potential}
The transport map can be written explicitly for any pair of one-dimensional distributions and in the Gaussian-to-Gaussian case in any dimension. The map may occasionally be constructed geometrically in special cases, such as in \citet{Back:1993:AsymmetricInformationOptions}. Thankfully, where an explicit optimal transport map is not known, numerical estimation is possible. Many algorithms exist, and efficient implementations are widely available. As observed in \citet{Back:1993:AsymmetricInformationOptions}, with options we expect discontinuities of the transport map $\nabla \Gamma(y)$. After exploring alternative methods including the Sinkhorn algorithm \citep{Cuturi:2013:SinkhornDistancesLightspeed},  input convex neural networks  \citep{makkuva2020optimal}, we chose to use classical linear programming approach for the computation of the transport map, because this method handles the discontinuities of the transport map better than alternatives. This methodology is based on generating a grid of size $N\times N$ (for $N\sim 151$) in the order imbalance space and a grid of size $M$ (with $M\sim 1001$) for the price of the underlying so that denoting $y_{i,j}\in \R^2$ and $v_k\in \R$ the grid points, for appropriate choice of $a$ and $b$, the measures
$$\sum_{i=1,j=1}^N a_{i,j} \delta_{y_{i,j}},\, \sum_{k=1}^M b_{k} \delta_{(v_k,(v_k-K)^+)}$$ 
approximate the two dimensional distributions of $Z_T$ and $\tilde v$. 
Then, the optimization problem with constraint 
$$\inf_{\pi} \sum_{i,j=1}^N\sum_{k=1}^M \pi_{i,j,k} \left(|y_{i,j,1}-v_k|^2+|y_{i,j,2}-(v_i-K)^+|^2\right)$$
is solved,\footnote{The POT (Python Optimal Transport) library contains an implementation of this procedure.} where $\pi$ is a probability distribution that has to satisfy positivity constraints and marginal distribution constraints given by $a_{i,j}$ and $b_k$. Thanks to Brenier's theorem the optimizer $\pi^*$ is supported on the graph of $\nabla \Gamma$. 

In order to obtain the potential $\Gamma$ from the values of $\nabla \Gamma(y_{i,j})$, we train a feed-forward neural network whose gradient at point $y_{i,j}$ is approximately $\nabla \Gamma(y_{i,j})$. This training step is done by minimizing the loss function
\begin{align}\label{eq:minpot}
\min_{f} \sum_{i=1,j=1}^N\sum_{k=1}^2\left|\frac{f(y_{i,j}+\e e_k)-f(y_{i,j}-\e e_k)}{2\e}- \pa_k \Gamma(y_{i,j})\right|^2+Conv(f)
\end{align}
where the minimization problem is over the parameters of the feed-forward neural network $f$ and $(e_k)_{k=1,2}$ is the canonical basis of $\R^2$. The term $Conv(f)$ is a cost for lack of convexity of the neural network that is defined in \citep{makkuva2020optimal}. The minimizer $f^*$ of this problem is an approximation of the potential function $\Gamma(y)$ and in this low dimensional setting we can plot $\nabla f$ and compare it to $\nabla \Gamma(y_{i,j})$.  If needed, we can compute $\mye[\Gamma(Z_T)]$ by a Monte Carlo scheme and subtract it from $\Gamma$ to obtain the normalization of $\Gamma$ in Theorem \ref{thm_brenier}.

\subsubsection{Backward integration of PDEs}
For time step $\Delta t>0$ fixed, the functions $\Gamma(t,y)$, $\nabla \Gamma(t,y)$ and $\E^{\phi(t,y)}$
satisfy the dynamic programming principles
$$\Gamma(t,y)= \mye\left[\Gamma(t+\Delta t,y+Z_{\Delta t})\right],\,\nabla\Gamma(t,y)= \mye\left[\nabla\Gamma(t+\Delta t,y+Z_{\Delta t})\right]$$
and 
\begin{align*}\E^{\phi(t,y)}&= \mye\left[\E^{\phi(t+\Delta t,y+Z_{\Delta t})+\alpha \tr\left(\Sigma \int_t^{t+\Delta t} \nabla^2 \Gamma(u,Z_u)\,\D u  \right) }\right]\\
&\sim \left(1+\alpha \Delta t \tr\left(\Sigma  \nabla^2 \Gamma(t,y)  \right)\right)\mye\left[\E^{\phi(t+\Delta t,y+Z_{\Delta t}) }\right]
\end{align*}
where the approximation is valid for small $\Delta t$. We choose $\Delta t\sim 1/100$. Then, using the transition densities of $Z_{\Delta t}$ on the grid $y_{i,j}$, we use the dynamic programming principles to compute these functions on the space-time grid $t=k \Delta t$
and $y=y_{i,j}$ backward starting from the final conditions 
$$\Gamma(T,y_{i,j})=\Gamma(y_{i,j}),\,\nabla \Gamma(T,y_{i,j})=\nabla \Gamma(y_{i,j}),\mbox{ and }\E^{\phi(T,y_{i,j})}=\E^{-\a y_{i,j}'\nabla \Gamma(y_{i,j})+\a \Gamma(y_{i,j})}.$$

\subsubsection{Computing the Physical Distribution}
We can take the logarithm of $\E^{\phi(k\Delta t,y_{i,j})}$ and differentiate numerically to obtain $\nabla \phi(k \Delta t,y_{i,j})$. Then, we simulate $Y$ under the physical measure as \eqref{newdY2} via an Euler scheme
$$Y_{k+1}=Y_k+\Sigma \nabla \phi(k\Delta t,Y_k)\,\Delta t + \Delta Z_{k+1}$$
where $\Delta Z_{k+1} = Z_{(k+1)\Delta t}-Z_{k\Delta t}$ and $\nabla \phi(k\Delta t,Y_k)$ is evaluated at $y_{i,j}$ closest to $Y_k$. The empirical distribution of the resulting random vector $\nabla \Gamma(Y_N)$ is approximately $\hat F$. 

$$E^P (\int (v-P_t)^2dt)^{1/2}=E^Q\xi (\int (v-P_t)^2dt)^{1/2}\leq (E^Q\xi^2)^{1/2}(E^Q (\int (v-P_t)^2dt))^{1/2}$$

%%%%%%%%%%%%%%%%%%%%%%%%%%%%%%%%%%%%%%%%%%%%%%%%%
%
% Example 3.1
%
%%%%%%%%%%%%%%%%%%%%%%%%%%%%%%%%%%%%%%%%%%%%%%%%%

\subsection{Example}\label{example:ralog}
Suppose there is a single risky asset, and it has a lognormal distribution under $F$ that is `winsorized' at some $v^*$, meaning that the probability mass above $v^*$ is shifted to $v^*$:  $\log \tv = m + \sigma_v (\tilde x \wedge x^*)$ for  constants $m$ and $\sigma_v$, where $\tilde x$ is a standard normal and $x^* = (\log  v^*-m)/\sigma_v$.   Because of the truncation, this example satisfies condition (i) of Theorem~\ref{thm:riskaverse2}.  Adopt the notation $\sigma_z^2 = \Sigma$ and $\lambda = \sigma_v/\sqrt{T}\sigma_z$.  The transport map is
$ \Gamma' (y) = \E^{m+ \lambda (y \wedge y^*)}$, where $y^* = \sqrt{T}\sigma_zx^*$.  
The Brenier potential is
\begin{align*}
    \Gamma(y) & = A + \frac{\E^m}{\lambda} \begin{cases}
\E^{\lambda y} & \text{if $y \le y^*$}\,,\\
\E^{\lambda y*} + \lambda (y-y^*)\E^{\lambda y*} & \text{if $y > y^*$}\,,
\end{cases}\\
& = A +\frac{\E^m}{\lambda}\bigg[\E^{\lambda y} - (\E^{\lambda y}-\E^{\lambda y^*})^+\bigg]  + \E^{m+\lambda y^*}( y-\ y^*)^+ \,,
\end{align*}
where the constant $A$ is determined by normalization.  
Thus, 
\begin{multline}
    \Gamma(t,y)= A+\frac{\E^m}{\lambda}\left[\E^{\lambda y+(T-t)\sigma_v^2/2T} -  C_{BS}\left(\E^{\lambda y+(T-t)\sigma_v^2/2T},\E^{\lambda y^*},\sigma_v/\sqrt{T},T-t\right)\right]\\
    +\E^{m+\lambda y*} \left[(y-y^*)N\left(\frac{y-y^*}{\sigma\sqrt{T-t}}\right)+\sigma\sqrt{T-t}N'\left(\frac{y-y^*}{\sigma_z\sqrt{T-t}}\right)\right]\,,
\end{multline}
where $N$ denotes the standard normal distribution function, $N'$ denotes the standard normal density function, and 
$C_{BS}(S,K,\sigma_v,\tau)$ is the Black-Scholes price of a call option with strike $K$ and time-to-maturity $\tau$ on an underlying with price $S$ and volatility $\sigma_v$.
The equilibrium price is
\begin{multline*}
   P_t = \frac{\partial\Gamma}{\partial y}(t,Y_t) = \E^{m+\lambda Y_t+(T-t)\sigma_v^2/2T} \left[1 - \Delta_{BS}\left(\E^{\lambda Y_t+(T-t)\sigma_v^2/2T},\E^{\lambda y^*},\sigma_v/\sqrt{T},T-t\right)\right]\\
    +\E^{m+\lambda y*} N\left(\frac{Y_t-y^*}{\sigma_z\sqrt{T-t}}\right)\,,
\end{multline*}
where $\Delta_{BS}(S,K,\sigma_v,\tau)$ denotes the call option delta.  Kyle's lambda is
\begin{align*}
   \frac{\partial^2\Gamma}{\partial y^2}(t,y) &= \lambda\E^{m+\lambda y+(T-t)\sigma_v^2/2T} \left[1 - \Delta_{BS}\left(\E^{\lambda y+(T-t)\sigma_v^2/2T},\E^{\lambda y^*},\sigma_v/\sqrt{T},T-t\right)\right]\\
   &- \lambda\E^{m}\E^{2\lambda y+(T-t)\sigma_v^2/T}  \gamma_{BS}\left(\E^{\lambda y+(T-t)\sigma_v^2/2T},\E^{\lambda y^*},\sigma_v/\sqrt{T},T-t\right)\\
    &+\frac{\E^{m+\lambda y*}}{\sigma_z\sqrt{T-t}} N'\left(\frac{y-y^*}{\sigma_z\sqrt{T-t}}\right)\,,
\end{align*}
where $\gamma_{BS}$ denotes the call option gamma.
%As in Corollary~\ref{thm:normal} and as is true in general, the  dynamics of $Y$ and $P$ are affected by risk aversion.  They are martingales in the risk-neutral model but not (under the physical probability) in the risk-averse model.  
To compute the physical dynamics with risk aversion, we need to compute $\partial \phi/\partial y$, which we can do by solving the PDE \eqref{pde}.  \begin{comment}
The definition \eqref{defphi} implies that
$$\exp\left(\phi(t,Z_t) + \alpha \sigma^2  \int_0^t  \Gamma_{yy}(u,Z_u)\,\D u\right)$$
is a $\nu$--martingale, where we use subscripts to denote partial derivatives. Assuming $\phi$ is sufficiently smooth to apply It\^o's formula, the martingality implies
\begin{equation}\label{phipde}
\phi_t + \frac{1}{2}\sigma^2\phi_{yy} + \frac{1}{2}\sigma^2 \phi_y^2 + \alpha \sigma^2  \Gamma_{yy}(t,y) = 0\,.
\end{equation}
We can find $\phi$ by solving this PDE in $(t,y)$ in conjunction with the boundary condition implied by \eqref{defphi}:
\begin{align}\phi(T,y) = \alpha \beta\E^{m+\lambda (y\wedge y^*)} - \alpha \Gamma^*\left(\E^{m+\lambda (y\wedge y^*)}\right)= \alpha \E^{m+\lambda (y\wedge y^*)}(\beta - y) + \a \Gamma(y)\,.\end{align}
\end{comment}
We then simulate the SDE
\begin{align}\label{eq:sde1}
\D Y_t = \sigma_z^2 \frac{\partial \phi}{\partial y}(t,Y_t)\,\D t + \D \hat Y_t \,,
\end{align}
generating $\hat Y$ as a $(0,\sigma_z^2)$--Brownian motion.  The physical distribution of $\tv$ is the distribution of $\E^{m+\lambda (Y_T \wedge y^*)}$ when $Y$ is simulated in this way (because $\tv = P_T = \E^{m+\lambda (Y_T\wedge y^*)}$ in equilibrium).  

Figure~A.1 illustrates the model.   The figure indicates that the risk-neutral variance of $\tv$ is higher than the physical variance.  We will establish this analytically for normal distributions in Section~\ref{s:normal}.  It occurs because the risk-neutral distribution overweights the tails of the distribution, where the market makers suffer losses to the informed trader.  In Panel~(a) of Figure~A.1, market makers start with a positive endowment of the asset.  Consequently, the risk-neutral distribution overweights low realizations of $\tv$, producing a risk-neutral mean that is lower than the physical mean, implying a positive risk premium.  The risk premium is negative in Panel~(b) where the endowment is negative.  We establish this relationship between the market makers' endowments and asset risk premia analytically in Section~\ref{s:normal} for normal distributions.  

%%%%%%%%%%%%%%%%%%%%%%%%%%%%%%%%%%%%%%%%%%%%%%%%%
%
% Figure 3.1
%
%%%%%%%%%%%%%%%%%%%%%%%%%%%%%%%%%%%%%%%%%%%%%%%%%

\begin{figure}\label{fig:32}
\begin{center}
\includegraphics[scale=0.6]{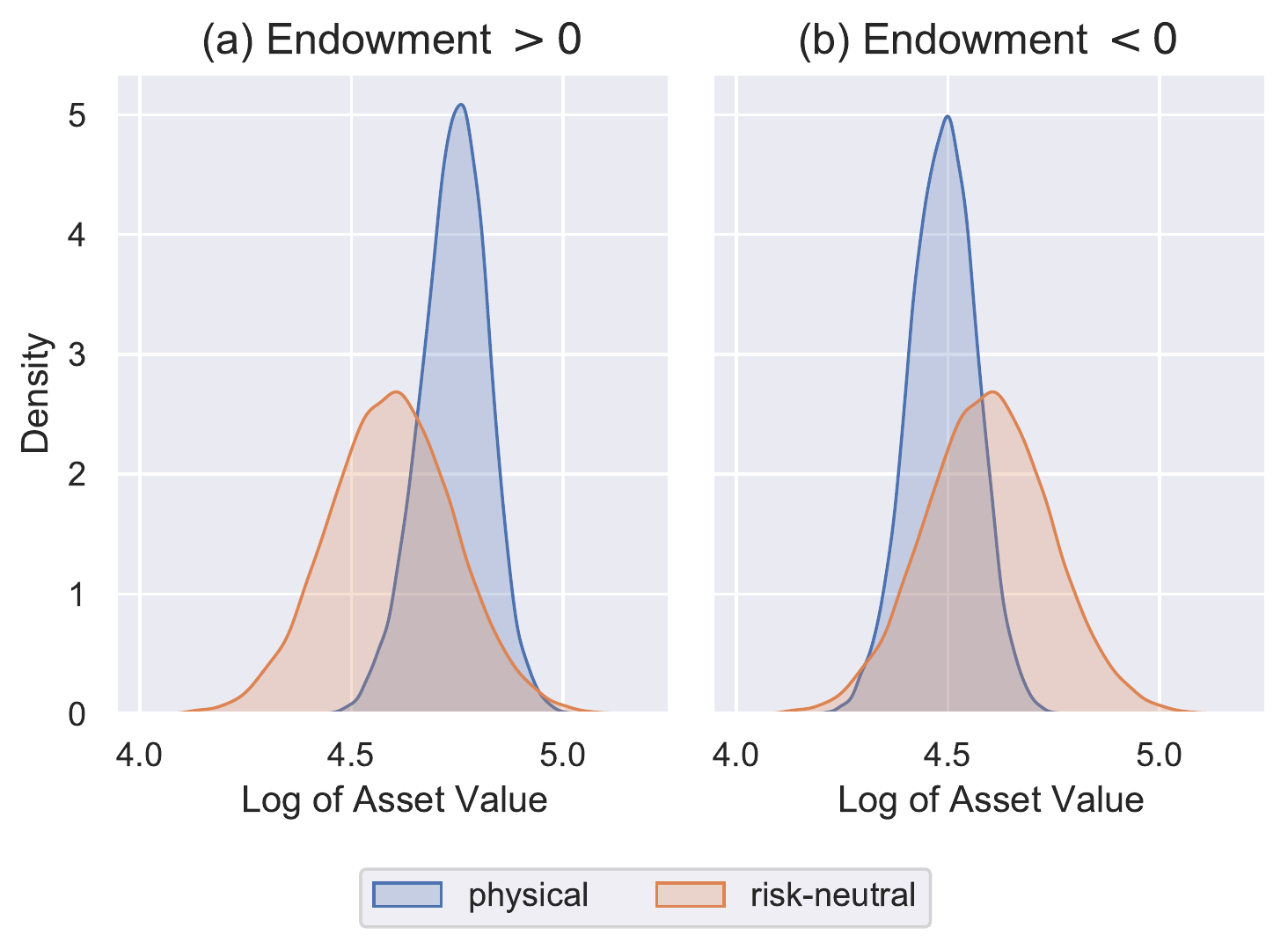}
\end{center}
\caption{\textbf{Physical and Risk-Neutral Distributions in a Truncated Lognormal Model} \newline
In both panels, the risk-neutral distribution is lognormal with mean $\,=100$ and standard deviation $\,=15$ winsorized at $\,160$.  Also, $\sigma_z\sqrt{T}=1$.  The endowment of market makers is $+1$ in Panel (a) and $-1$ in Panel (b).  Risk aversion is $\alpha=0.2$.}
\end{figure}

We can follow the same procedure to solve a lognormal model without truncation, though we have been unable to verify that the non-truncated model satisfies the assumptions of Theorem~\ref{thm:riskaverse}.  Without truncation, the potential is 
$\Gamma(y) = \E^{m+\lambda y}/\lambda-\E^{m+\sigma_v^2/2}/\lambda$,  and
$$\Gamma(t,y) = \frac{1}{\lambda}\E^{m+\lambda y + (T-t)\sigma_v^2/2T}-\E^{m+\sigma_v^2/2}/\lambda\,.$$
The transport map is $ \Gamma' (y) = \E^{m+\lambda y}$, and
the price is
$$P_t = \Gamma_y(t,Y_t) = \E^{m+\lambda Y_t + (T-t)\sigma_v^2/2T}\,.$$
Thus, $\D P_t/P_t = \lambda\,\D Y_t$.
%The PDE for $\phi$ is
%$$
%\phi_t + \frac{1}{2}\sigma^2\phi_{yy} + \frac{1}{2}\sigma^2 \phi_y^2 + \alpha \sigma^2 \lambda \E^{m+\lambda y + (T-t)\sigma_v^2/2T} = 0\,.
%$$

\end{appendix}

\end{document}